\documentclass[aps, pra, english,superscriptaddress,floatfix,twocolumn]{revtex4-2}

\usepackage{bm}
\usepackage[utf8]{inputenc}
\usepackage{amssymb}
\usepackage{amsthm}
\usepackage{amsmath}
\usepackage{thmtools, thm-restate}

\usepackage[lined,algoruled,algosection,linesnumbered,noend]{algorithm2e}
\SetAlCapSkip{0.5em}
\SetEndCharOfAlgoLine{} 
\SetKwInOut{Data}{Data}
\SetKwInOut{Input}{Input}
\SetKwInOut{Output}{Out}
\SetKwFor{For}{for}{\string:}{endfor}
\SetKwFor{While}{while}{\string:}{endw}
\SetKwFor{ForEach}{foreach}{\string:}{endfch}
\SetKwFor{ParFor}{parallel for}{\string:}{endpfor}
\SetKwProg{Fn}{function}{\string:}{endfunction}

\usepackage{tikz}
\usetikzlibrary{hobby}

\usepackage{qcircuit}

\usepackage{xcolor}
\usepackage[colorlinks=true,citecolor=blue,linkcolor=purple]{hyperref}

\usepackage[caption=false]{subfig}

\usepackage{mathtools}
\mathtoolsset{centercolon}

\DeclareMathOperator*{\Ex}{\mathbb{E}}
\DeclareMathOperator*{\Prob}{Pr}
\DeclareMathOperator{\traceOp}{tr}
\DeclareMathOperator{\spanOp}{span}
\DeclareMathOperator{\RealOp}{Re}

\DeclareMathOperator*{\selectOp}{select}

\DeclarePairedDelimiter\ceil{\lceil}{\rceil}

\DeclarePairedDelimiter\bra{\langle}{\vert}
\DeclarePairedDelimiter\ket{\vert}{\rangle}
\DeclarePairedDelimiterX\braket[2]{\langle}{\rangle}%
{#1\delimsize\vert\mathopen{}#2}

\DeclarePairedDelimiter\paren{(}{)}
\DeclarePairedDelimiter\bracket{[}{]}
\DeclarePairedDelimiter\abs{\lvert}{\rvert}
\DeclarePairedDelimiter{\norm}{\lVert}{\rVert}
\DeclarePairedDelimiterX\brakett[3]{\langle}{\rangle}%
{#1\delimsize\vert\,\mathopen{}#2\,\delimsize\vert\mathopen{}#3}
\DeclarePairedDelimiterX\ketbra[2]{\vert}{\vert}{#1 {\delimsize\rangle\langle} #2}
\DeclarePairedDelimiterX{\inp}[2]{\langle}{\rangle}{#1, #2} 
\DeclarePairedDelimiterXPP\bigo[1]{O}{(}{)}{}{#1}
\DeclarePairedDelimiterXPP\littleo[1]{o}{(}{)}{}{#1}
\DeclarePairedDelimiterXPP\bigomega[1]{\Omega}{(}{)}{}{#1}
\DeclarePairedDelimiterXPP\bigtheta[1]{\Theta}{(}{)}{}{#1}
\DeclarePairedDelimiterXPP\probability[1]{\Prob}{[}{]}{}{#1}
\DeclarePairedDelimiterXPP\probabilityq[2]{\underset{#1}{\Prob}}{[}{]}{}{#2}
\DeclarePairedDelimiterXPP\expectation[1]{\Ex}{[}{]}{}{#1}
\DeclarePairedDelimiterXPP\expectationq[2]{\Ex_{#1}}{[}{]}{}{#2}
\DeclarePairedDelimiterXPP\entropy[1]{\entropyOp}{(}{)}{}{#1}
\DeclarePairedDelimiterXPP\information[1]{\informationOp}{(}{)}{}{#1}
\DeclarePairedDelimiterXPP\trace[1]{\traceOp}{[}{]}{}{#1}
\DeclarePairedDelimiterXPP\ptrace[2]{\traceOp_{#1}}{[}{]}{}{#2}
\DeclarePairedDelimiterXPP\spanSet[1]{\spanOp}{\{}{\}}{}{#1}
\DeclarePairedDelimiterXPP\fidelity[2]{\fidelityOp}{(}{)}{}{#1, #2}
\DeclarePairedDelimiterXPP\fidelitywc[2]{\fidelityWCOp}{(}{)}{}{#1, #2}
\DeclarePairedDelimiterXPP\Real[1]{\RealOp}{(}{)}{}{#1}
\DeclarePairedDelimiterXPP\select[1]{\selectOp}{(}{)}{}{#1}

\DeclareMathOperator{\poly}{poly}

\newcommand{\swap}{swap}
\newcommand{\li}{l}
\newcommand{\ri}{r}
\newcommand{\lr}[1]{\left( #1\right)}

\newcommand{\tU}{\widetilde{U}}
\newcommand{\ii}{\mathrm{i}}
\newcommand{\dd}{\mathrm{d}}
\newcommand{\spliteq}{\nonumber \\&\qquad}
\newcommand{\rt}[1]{\mathrm{rt}(#1)}
\newcommand{\hrt}[1]{\mathrm{hrt}(#1)}
\newcommand{\hrtd}[1]{\mathrm{hrt}_{\Delta}(#1)}

\newcommand{\rtf}[1]{\mathrm{rt}_F(#1)}
\newcommand{\hrtf}[1]{\mathrm{hrt}_F(#1)}

\newcommand{\Acoeffvec}{A_{\overrightarrow{\li_{p_1}},\overrightarrow{\ri_{q}}, \overrightarrow{c_{p_2}}; P,Q}^{\overrightarrow{\alpha_{p_1}}, \overrightarrow{\beta_{q}},\overrightarrow{\gamma_{p_2}}}}
\newcommand{\Acoeff}{A_{\li_1\cdots \li_{p_1}, \ri_1 \cdots \ri_{q}, c_1 \cdots c_{p_2}; P,Q}^{\alpha_1\cdots \alpha_{p_1}, \beta_1\cdots \beta_{q},\gamma_1 \cdots \gamma_{p_2}}}
\newcommand{\Atilde}{A_{\widetilde{\li_1}\cdots \widetilde{\li_{\widetilde{p_1}}}, \widetilde{c_1} \cdots \widetilde{c_{\widetilde{p_2}}},\ri_1\cdots \ri_{q}; P+1,Q}^{\widetilde{\alpha_1}\cdots \widetilde{\alpha_{\widetilde{p_1}}}, \widetilde{\beta_1}\cdots \widetilde{\beta_{q}},\widetilde{\gamma_1} \cdots \widetilde{\gamma_{\widetilde{p_2}}}}}

\newcommand{\Cavg}{\expectationq*{\ket{\psi} \sim \mu}{C_{\Delta}(\ket{\psi},t)}}

\newcommand{\hc}{\mathrm{h.c.}}

\usepackage[capitalise,compress,nameinlink]{cleveref}
\crefname{section}{Sec.}{Secs.}
\crefrangelabelformat{equation}{(#3#1#4)--(#5#2#6)}

\usepackage{apptools}
\AtAppendix{\counterwithin{lemma}{section}}
\AtAppendix{\counterwithin{theorem}{section}}
\AtAppendix{\counterwithin{definition}{section}}
\newtheorem{theorem}{Theorem}[section]

\newtheorem{lemma}[theorem]{Lemma}
\newtheorem{corollary}[theorem]{Corollary}

\theoremstyle{definition}
\newtheorem{definition}[theorem]{Definition}

\makeatletter
\renewcommand{\p@subsection}{}
\renewcommand{\p@subsubsection}{}
\makeatother

\begin{document}

\title{Quantum Routing and Entanglement Dynamics Through Bottlenecks}
\author{Dhruv Devulapalli}
\thanks{D. D. and C. Y. contributed equally. Corresponding authors: devulapallidhruv@gmail.com, chao.yin@colorado.edu}
    \affiliation{Joint Center for Quantum Information and Computer Science, 
        NIST/University of Maryland, College Park, MD 20742, USA}
    \affiliation{Joint Quantum Institute, NIST/University of Maryland, 
        College Park, MD 20742, USA}

\author{Chao Yin}
\thanks{D. D. and C. Y. contributed equally. Corresponding authors: devulapallidhruv@gmail.com, chao.yin@colorado.edu}
\affiliation{Department of Physics and Center for Theory of Quantum Matter,
University of Colorado, Boulder, CO 80309, USA}

\author{Andrew~Y.~Guo}
\affiliation{Quantinuum, Broomfield, CO, 80021, USA}

\author{Eddie~Schoute}
    \affiliation{Computer, Computational, and Statistical Sciences Division, Los Alamos National Laboratory, Los Alamos, NM 87545, USA}
    \affiliation{IBM Quantum, MIT-IBM Watson AI Lab, Cambridge, MA 02142, USA}
        
\author{Andrew~M.~Childs}
    \affiliation{Joint Center for Quantum Information and Computer Science, 
        NIST/University of Maryland, College Park, MD 20742, USA}
    \affiliation{Institute for Advanced Computer Studies, University of Maryland, College Park,
        MD 20742, USA}
    \affiliation{Department of Computer Science, University of Maryland, College Park,
        MD 20742, USA}
        
\author{Alexey~V.~Gorshkov}
    \affiliation{Joint Center for Quantum Information and Computer Science, 
        NIST/University of Maryland, College Park, MD 20742, USA}
    \affiliation{Joint Quantum Institute, NIST/University of Maryland, 
        College Park, MD 20742, USA}

\author{Andrew~Lucas}   \affiliation{Department of Physics and Center for Theory of Quantum Matter,
University of Colorado, Boulder, CO 80309, USA}

\begin{abstract}
    To implement arbitrary quantum circuits in architectures with restricted interactions, one may effectively simulate all-to-all connectivity by routing quantum information. We consider the entanglement dynamics and routing between two regions only connected through an intermediate ``bottleneck'' region with few qubits. In such systems, where the entanglement rate is restricted by a vertex boundary rather than an edge boundary of the underlying interaction graph, existing results such as the small incremental entangling theorem give only a trivial constant lower bound on the routing time (the minimum time to perform an arbitrary permutation). We significantly improve the lower bound on the routing time in systems with a vertex bottleneck. Specifically, for any system with two regions $L, R$  with $N_L, N_R$ qubits, respectively, coupled only through an intermediate region $C$ with $N_C$ qubits, for any $\delta > 0$ we show a lower bound of $\bigomega{N_R^{1-\delta}/\sqrt{N_L}N_C}$ on the Hamiltonian quantum routing time when using piecewise time-independent Hamiltonians, or time-dependent Hamiltonians subject to a smoothness condition. We also prove an upper bound on the average amount of bipartite entanglement between $L$ and $C,R$ that can be generated in time $t$ by such architecture-respecting Hamiltonians in systems constrained by vertex bottlenecks, improving the scaling in the system size from $\bigo{ N_L t}$ to $\bigo{ \sqrt{N_L} t}$. As a special case, when applied to the star graph (i.e., one vertex connected to $N$ leaves), we obtain an $\bigomega{\sqrt{N^{1-\delta}}}$ lower bound on the routing time and on the time to prepare $N/2$ Bell pairs between the vertices. We also show that, in systems of free particles, we can route optimally on the star graph in time $\bigtheta{\sqrt{N}}$ using Hamiltonian quantum routing, obtaining a speed-up over gate-based routing, which takes time $\bigtheta{N}$.
\end{abstract}

\maketitle

\section{Introduction}
The promise of wide-ranging applications of quantum computing has sparked interest in developing scalable quantum architectures. 
Realistic quantum architectures have constrained interactions, often with fixed connectivity~\cite{BCMS19, Kjaergaard20}. This connectivity is usually specified by a graph $G$, with vertices $V$ representing qubits, and edges $E$ representing pairs of qubits between which interactions are permitted. On the other hand, quantum algorithms, quantum error correction schemes, and other quantum information-processing protocols often require interactions between arbitrary pairs of qubits. To implement these interactions, quantum compilers effectively simulate all-to-all connectivity by moving quantum information,
at the cost of an overhead in the running time. This task, known as quantum routing, aims to make general quantum algorithms feasible on devices with connectivity constraints by efficiently implementing any permutation of qubits on a given architecture graph~\cite{Childs2019}.
Quantum routing can be performed with \swap{} gates using classical techniques \cite{Alon1994,cowtan, Wagner_Bärmann_Liers_Weissenbäck_2023, Childs2019, Li_Ding_Xie_2019}, or employing more general quantum operations such as continuous-time Hamiltonian evolution \cite{reversals,adv_lim} or unitary gates \cite{adv_lim}, possibly assisted by measurement and fast classical feedback \cite{teleportation, rosenbaum, Hahn_Pappa_Eisert_2019}.
Recent work has shown that the worst-case depth overhead of implementing arbitrary quantum circuits while respecting connectivity constraints is proportional to the circuit depth of worst-case routing \cite{yuan2024}, highlighting the importance of lower bounds and fast protocols for this task.

 Routing is closely related to the task of entanglement distribution, since by routing entangled qubits, we can distribute entanglement. 
The rate at which entanglement can be generated between parts of a system constrains our ability to perform many quantum information processing tasks, as entanglement is a critical resource for quantum algorithms \cite{Jozsa_1997, Jozsa_Linden_2003, Ekert_Jozsa_1998, Bruß_Macchiavello_2011, Zhao_Zhou_Childs_2024}, quantum error correction \cite{Bravyi_Lee_Li_Yoshida_2024, Bennett_DiVincenzo_Smolin_Wootters_1996, Knill_Laflamme_1997}, quantum communication \cite{Buhrman_communication, Gisin_Thew_2007, Horodecki_2009}, and quantum sensing \cite{Degen_Reinhard_Cappellaro_2017, Giovannetti_Lloyd_Maccone_2004, Bollinger_Itano_Wineland_Heinzen_1996}. The fundamental importance of entanglement and its application to a wide range of quantum information tasks has raised many questions and led to a large number of investigations into the dynamics of entanglement \cite{SIE, STE, flow, edss, bravyi_upper_bounds, Dür_Vidal_Cirac_Linden_Popescu_2001, Bennett_Harrow_Leung_Smolin_2003}. 

The fact that routing can be used to distribute entanglement means that bounds on entanglement dynamics can also be used to constrain the time taken to perform routing. Previous work has used this idea to bound routing times \cite{adv_lim} in terms of entanglement rates \cite{SIE} and entanglement capacities \cite{STE} across system bipartitions. These two major results on entanglement dynamics can be interpreted in terms of bottlenecks, which are of two kinds:
\begin{enumerate}
    \item Edge bottlenecks: The number of edges across any bipartition of the system constrains the rate at which the bipartitions can be entangled, by the small incremental entangling (SIE) theorem \cite{SIE}. Thus, in systems with a small number of edges connecting large subsystems, the routing time is large. An example of such system is a $d$-dimensional array, where the number of edges joining two halves of the system scales as $\bigtheta{N^{d-1}}$, which is asymptotically smaller than the number of vertices on each side ($\bigtheta{N^d}$). We refer to such a set of edges as an edge bottleneck.
    \item Vertex bottlenecks: A vertex bottleneck is a set of vertices, typically much smaller than the rest of the system, that constrains the interactions of the rest of the system. An archetypal example of a vertex bottleneck is found in the star graph $S_{N}$, shown in \cref{fig:star_graph}, 
    \begin{figure}[t]
    \centering
    \includegraphics[width=0.5\columnwidth]{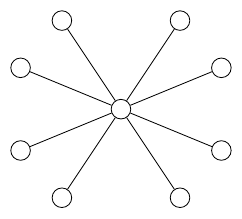}
    \caption{The star graph on 9 vertices, $S_9$.}
    \label{fig:star_graph}
\end{figure}
    which is the complete bipartite graph $K_{1,N-1}$ between one (center) vertex and $N-1$ leaf vertices. Intuitively, to implement interactions between qubits on the leaves, we must make use of the central vertex, which forms a bottleneck. On the other hand, this graph does not have a significant edge bottleneck, due to the large number of edges between the central vertex and the leaves. A general vertex bottleneck is depicted in \cref{fig:tripartite}, 
    \begin{figure}[t]
    \centering
    \includegraphics[width=\columnwidth]{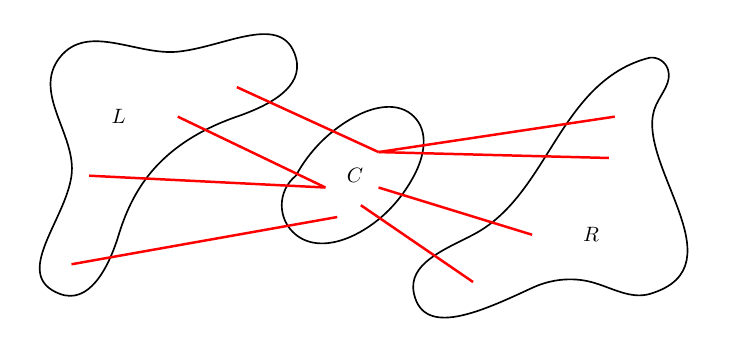}

    \caption{A tripartition with vertex bottleneck $C$. Parts $L, C, R$ contain $N_L, N_C, N_R$ qubits, respectively. By the tripartite connectivity constraint, there are no edges between qubits in $L$ and qubits in $R$.}
    \label{fig:tripartite}
\end{figure}
    which represents systems with a tripartite structure. Such systems can be divided into three parts $L, C, R$ such that edges are only present connecting $L$ to $C$ and $C$ to $R$. The small total entangling property \cite{STE} bounds the entanglement capacity, or total amount of entanglement that can be generated across a partition, in terms of the local dimension of the vertex boundary (sites on one side of a partition that are adjacent to sites on the other side), and therefore constrains the ability to generate entanglement through a vertex bottleneck.
\end{enumerate}
In the gate-based routing model, we make use of 2-qubit gates that act locally, obeying the connectivity constraints of the architecture. A special case of this is classical swap-based routing, where only connectivity-respecting swap gates are permitted. Alternatively, in the more powerful Hamiltonian routing model \cite{adv_lim}, we are permitted continuous-time evolution by an architecture-respecting Hamiltonian with norm-bounded interactions. In this work, we primarily investigate piecewise time-independent Hamiltonians.
The routing time of a model for a given graph $G$ is defined as the time taken to implement a worst-case permutation on $G$. In the gate-based model, this corresponds to the minimum circuit depth to implement the permutation, while in the Hamiltonian routing model this corresponds to the minimum evolution time required. These routing models are described in more detail in \cref{sec:prelims}.  

Using bounds on the entangling rate \cite{SIE}, tight bounds can be obtained for routing constrained by edge bottlenecks in both the gate-based and continuous-time evolution routing models. However, for routing through vertex bottlenecks, only entanglement bounds in the gate-based model \cite{STE} have been found to be tight~\cite{adv_lim}, following the natural assumption that no two gates can interact with the same qubit at the same time. Attempts at a full treatment of the entanglement dynamics in multipartite systems reveal surprisingly more complexity than in the bipartite case. For example, \textcite{edss} showed that in a system with a tripartition, entanglement can be generated between the two separated parts without ever entangling the intermediate region with the rest of the system. While this result defies the natural intuition that entanglement must flow between regions of a system, such separable entanglement generation is only possible if the initial state is a special mixed state. 
For systems starting in a pure state, a notion of entanglement flow between different subsystems was shown by \textcite{flow}. However, this result still provides only a trivial lower bound for routing through vertex bottlenecks.
 
For the star graph, the gate-based quantum routing time is $\bigtheta{N}$~\cite{adv_lim}.
However, in the Hamiltonian routing model,
which has previously only been shown to be constrained by edge bottlenecks,
the best known lower bound on the routing time is $\bigomega{1}$ (applying the result of \textcite{SIE} or a Lieb-Robinson bound \cite{QSLreview23}).  
This lower bound intuitively seems loose, since if it were indeed saturable,
this would allow for routing in constant time, independent of the system size.

Surprisingly, polynomial speedups for continuous-time evolution over the gate-based model are indeed possible, even in the simple setting of the star graph. For one of our main results, we investigate systems of free (i.e., non-interacting)  particles, where the edges of the underlying graph represent two-body hopping terms in the Hamiltonian. In systems of free particles, the gate-based model can be considered to allow architecture-respecting circuits of free-particle gates (i.e., unitaries generated by hopping and on-site interactions). The gate-based routing time on the star graph for free particles is $\bigtheta{N}$, by the Small Total Entangling lemma \cite{STE}, applied to circuits of free-particle gates. For these systems, we design a fast Hamiltonian routing protocol that takes time $t = \bigo{\sqrt{N}}$ for arbitrary permutations, which demonstrates a $\bigomega{\sqrt{N}}$ speedup over the gate-based model.
We further provide a concise proof that this protocol is optimal (for both free fermions and free bosons), showing a lower bound of $t = \bigomega{\sqrt{N}}$ on the routing time on the star graph for these systems. We note that this lower bound applies even to time-dependent Hamiltonian routing.

Returning to the more standard setting of qubit systems,
we significantly tighten the lower bound on the Hamiltonian routing time for the star graph to $\bigomega{N^{1/2-\delta}}$ for any $\delta>0$, when using piecewise time-independent Hamiltonians, or time-dependent Hamiltonians subject to a smoothness condition. This
addresses an open question of \textcite{adv_lim}. 
The first part of our proof leverages improved commutator bounds on Hamiltonian simulation for Trotter-Suzuki formulas  \cite{trotter_comm21} for random input states \cite{Trotter_Frob22}. This result utilizes the Frobenius norm, allowing us to leverage the fact that vertex-bottlenecked Hamiltonians exhibit tighter commutator scaling in the Frobenius norm compared to the operator norm used in entanglement dynamics bounds \cite{SIE} or Lieb-Robinson bounds \cite{QSLreview23}. In contrast to those results, we take advantage of the inherently state-independent task of performing large-scale routing to bound the \textit{average}-case routing time rather than just the worst case, which allows a more precise bound based on the Frobenius norm.  The technique behind our proof is inspired by the recent application of ``Frobenius light cone" bounds, which demonstrated the hardness of implementing the shift-by-one translation on a ring architecture \cite{shift}.

This result generalizes to arbitrary vertex bottlenecks, as illustrated in \cref{fig:tripartite}. For an arbitrary tripartition of $N_L, N_C, N_R$ qubits (with $N_L \geq N_R \geq N_C$), we show that the routing time is $\bigomega{N_R^{1-\delta}/\sqrt{N_L}N_C}$ for any $\delta>0$.

Our techniques further allow us to prove an upper bound on the entanglement capacity, or the amount of entanglement that can be generated or distributed through a vertex bottleneck. For a tripartition, we define the entanglement capacity for evolution of a state $\ket{\psi}$ by a given Hamiltonian $H$ for time $t$ as the increase in the amount of bipartite entanglement between the $L$ subsystem and the rest of the system. Existing results \cite{SIE, flow} bound the entanglement capacity as $\bigo{tN_LN_C}$ for any initial state $\ket{\psi}$. For the case where the input state is chosen from a 1-design ensemble, we show an improved upper bound of $\bigo{tN_C\sqrt{N_L} (\sqrt{N_LN_R/N_C} )^{\delta}}$, for any $\delta > 0$, on the entanglement capacity in systems with a tripartition. By averaging over input states, we obtain a tighter bound that depends on Frobenius norms of commutators between terms of the Hamiltonian, rather than the operator norm of the Hamiltonian as used in Refs.~\cite{SIE, flow}. We believe this result could be useful in bounding the time taken for more general quantum information processing tasks beyond routing.  For example, our methods may be useful in demonstrating the hardness of implementing specific unitary gates on a quantum processor. 

The remainder of the paper is organized as follows. In \cref{sec:prelims}, we introduce our models and technical background for the routing problem.  In \cref{sec:free_protocol}, we describe a fast routing protocol for systems of free fermions with star-graph connectivity, which performs routing in time $\bigo{\sqrt{N}}$. In \cref{sec:qubit_nogo} we discuss the challenges of converting this to a qubit routing protocol. In \cref{sec:free_bound}, we show a concise lower bound of $\bigomega{\sqrt{N}}$ on the routing time for fermions in systems with a vertex bottleneck, including the star graph, which demonstrates that our protocol in \cref{sec:free_protocol} is optimal. In \cref{sec:qubit_bound}, we show our main result bounding the routing time through vertex bottlenecks in systems of qubits. In \cref{sec:ent_cap}, we show how similar techniques provide a lower bound on the average entangling capacity in systems with a tripartition.  Finally, we discuss implications of our results and some open questions in \cref{sec:discussion}.

\section{Preliminaries}
\label{sec:prelims}

We let $\norm{\cdot}$ denote the spectral/operator norm and $\norm{\cdot}_p$ the Schatten $p$-norm. The (normalized) Frobenius norm of an operator $O$ is
\begin{equation}\label{eq:Frob}
    \norm{O}_{\rm F} \coloneqq \sqrt{\frac{\trace{O^\dagger O}}{D}},
\end{equation}
where $D$ is the total Hilbert space dimension. Note that $\norm{O}_{\rm F} = \norm{O}_2/\sqrt{D}$. With this normalization, $\norm{\mathbf{X}}_{\rm F}=1$ for any Pauli string $\mathbf{X}$. 

We specify the connectivity of a quantum system with 2-local interactions in the form of a graph $G=(V,E)$, with the vertices $V$ representing qubits/modes and the edges $E$ representing pairs of qubits/modes between which interactions are allowed.
In this paper, we are mostly concerned with graphs with a vertex bottleneck, as depicted in \cref{fig:tripartite}. 
Routing is the task of implementing arbitrary permutations from $\mathcal{S}_N$, the group of all permutations of $N$ elements, on qubits labeled from $1$ to $N$.
For any $p\in \mathcal{S}_N$, where $p(i)$ represents the destination of the $i^{\mathrm{th}}$ qubit, we define the permutation unitary $U_p$, such that, for any product state $\ket{\psi} = \ket{\psi_1}\ket{\psi_2}\dots  \ket{\psi_N}$,
\begin{equation}
    U_p \ket{\psi_1}\ket{\psi_2}\dots  \ket{\psi_N} = \ket{\psi_{p(1)}}\ket{\psi_{p(2)}}\dots  \ket{\psi_{p(N)}}.
\end{equation}

\subsection{Hamiltonian routing}

In the Hamiltonian routing model, we consider continuous-time evolution by a Hamiltonian $H$ that respects the connectivity constraints. Such an architecture-respecting evolution has no interaction terms between sites (qubits) that are not connected by an edge in the underlying graph $G$. Thus, any valid Hamiltonian can be written as 
\begin{equation}
    H = \sum_{(i,j) \in E}h_{ij} + \sum_{i\in V}h_i,
\end{equation}
where $h_{ij}$ is only supported on qubits $i,j$ (i.e., acts as identity on all other qubits). We further impose the constraint $\norm{h_{ij}} \leq 1$ in order to ensure that the timescale of two-site interactions is comparable to the time required for a 2-qubit gate in the gate-/swap-based routing models.
Unitary evolution by a Hamiltonian $H$ is given by the time-evolution operator
\begin{equation}
    U(H,t) \coloneqq \mathcal{T} e^{-\ii \int_0^t H(t')\,\dd t'},
\end{equation}
where $\mathcal{T}$ is the time-ordering operator. When $H$ is time-independent, $U(H,t) = e^{-\ii Ht}$. 

The \emph{Hamiltonian routing time} is the minimum time to implement a given permutation in this model, defined as 
\begin{equation}
    \hrt{G,p} \coloneqq \min_H \{t \textrm{ s.t. } U(H,t) = U_p\}
\end{equation}
for any $p \in \mathcal{S}_N$.
We define the worst-case Hamiltonian routing time for a graph $G$ as
\begin{equation}
    \hrt{G} \coloneqq \max_{p \in \mathcal{S}_{N}} \hrt{G,p}.
\end{equation}

In this work, we primarily consider time-independent or piecewise time-independent $H$.
\begin{definition}[Piecewise time-independent Hamiltonian]
\label{def:piecewise}
    $H$ is said to be piecewise time-independent with minimum segment width $\Delta$ if it can be written as 
    \begin{equation} 
    H = \sum_{i=0}^{k-1} H_i \cdot \mathbb{I}_{t\in [t_i, t_{i+1})},
    \end{equation}
    where $t_k = t$, $H_i$ is time-independent, $\mathbb{I}_{t\in [t_i, t_{i+1})}$ is an indicator function for ${t\in [t_i, t_{i+1})},$ and $\forall \, i,t_{i+1}-t_i \geq \Delta$.
\end{definition}

Similarly to the time-dependent case, we define the Hamiltonian routing time for piecewise time-independent Hamiltonians with minimum segment width $\Delta$ as 
\begin{equation}
    \hrtd{G,p} \coloneqq \min_H \{t \textrm{ s.t. } U(H,t) = U_p\}
\end{equation}
for any $p \in \mathcal{S}_N$.
We define the worst-case Hamiltonian routing time for piecewise time-independent Hamiltonians with minimum segment width $\Delta$ for a graph $G$ as
\begin{equation}
    \hrtd{G} \coloneqq \max_{p \in \mathcal{S}_{N}} \hrtd{G,p}.
\end{equation}

\subsection{Gate-based routing}

The gate-based routing model is a variant of the Hamiltonian routing model where we are restricted to performing 1-qubit and 2-qubit gates that respect the connectivity constraints (i.e., no gates are permitted between qubits not connected by an edge). Gates on disjoint pairs of vertices may be applied simultaneously. In this model, the routing time for a given permutation is the minimum depth to implement a given permutation $p$ by an architecture-respecting circuit,
\begin{equation}
    \rt{G,p} \coloneqq \min_U \textrm{depth}(U) \textrm{ s.t. } U = U_p.
\end{equation} 
Similar to the Hamiltonian routing model, we also define the \emph{worst-case routing time} for a graph $G$ as
\begin{equation}
    \rt{G} \coloneqq \max_{p \in \mathcal{S}_{N}} \rt{G,p}.
\end{equation}

The gate-based model generalizes classical routing, which is restricted to circuits composed of architecture-respecting swap gates \cite{Alon1994}.

\subsection{Free-particle routing}
An alternative type of Hamiltonian routing is free-particle routing (of bosons or fermions).

In a fermionic model, each vertex of $G$ represents a local fermionic mode, which can be either empty $(\ket{0})$ or occupied by a fermion $(\ket{1})$. Fermionic states and interactions are represented in terms of creation and annihilation operators on each mode $(c_j^{\dagger}, c_j)$ \cite{bravyi2002fermionic}.

The annihilation operator acts as follows:
\begin{subequations}\begin{align}
    &c_j \ket{n_0, n_1, \dots, n_{j-1}, 1, n_{j+1}, \dots, n_N} \nonumber \\
    &\qquad= (-1)^{\sum_{k=0}^{j-1} n_k}\ket{n_0, n_1, \dots, 0_j, \dots, n_N}, \\
    &c_j \ket{n_0, n_1, \dots, n_{j-1}, 0, n_{j+1}, \dots, n_N} = 0.
\end{align}\end{subequations}
These operators obey the fermionic anticommutation relations
\begin{equation}
    \{c_j, c_k \} = \{c_j^{\dagger}, c_k^{\dagger} \} = 0, ~ \{c_j, c_k^{\dagger}\} = \delta_{j,k} .
\end{equation}

We also consider routing in systems of free (non-interacting) bosons. Similarly to the fermionic case, states are characterized as Fock states, i.e., in the occupancy number basis $\ket{n_i}$ where $n_i$ is the occupancy of the $i^{th}$ mode (situated on  site $i$). The bosonic creation and annihilation operators are $(a_j^{\dagger}, a_j)$. The annihilation operator acts as follows:
\begin{align}
    &a_j \ket{n_0, n_1, \dots, n_{j-1}, n_j, n_{j+1}, \dots, n_N} \nonumber \\
    &\qquad=  \sqrt{n_j} \ket{n_0, n_1, \dots, n_j-1, \dots, n_N}. 
\end{align}

These operators obey the bosonic commutation relations
\begin{equation}
    [a_j, a_k ] = [a_j^{\dagger}, a_k^{\dagger} ] = 0, ~ [a_j, a_k^{\dagger}] = \delta_{j,k} .
\end{equation}

Any fermionic or bosonic Hamiltonian can be written as a sum of products of creation and annihilation operators. An architecture-respecting fermionic/bosonic Hamiltonian can be written as a sum of products of these operators on sites that are connected by an edge in $G$, and have bounded coefficients. In a free (non-interacting) model, the Hamiltonian only contains hopping and on-site terms (in \cref{sec:discussion}, we comment on the possibility of allowing for pairing). 
In particular, a free-particle Hamiltonian can be written as
\begin{equation}
    \label{eq:free_fermionic_H}
    H = \sum_{i \in V} h_i c_i^{\dagger}c_i + \sum_{(i, j) \in E} \left(h_{ij} c_i^{\dagger}c_j + \hc\right),
\end{equation}
where $|h_{ij}| \leq 1$, $h_i, h_{ij}$ may have arbitrary time dependence, and $c_i$ are all fermionic or all bosonic.

A permutation $p \in \mathcal{S}_N$ maps each mode on site $i$ to site $p(i)$. Therefore, $U_p$ maps the set of annihiliation operators $\{c_1, c_2, \dots\}$ to $\{c_{p(1)}, c_{p(2)}, \dots\}$.
For particles evolving by a free Hamiltonian with modes $b_i$, the time-evolved creation and annihilation operators at a given time $t$ can be written as $b_i(t) = \sum_{j} A_{ij}(t) b_j(0)$ and $b^\dagger_i(t) = \sum_{j} A_{ij}^*(t) b^\dagger_j(0)$ for some $A_{ij}(t) \in \mathbb{C}$. The dynamics of $A_{ij}(t)$ are the same for both bosons and fermions (and for a single particle). Hence, any protocol that performs routing for free fermions yields an identical routing protocol for free bosons and vice versa. 

We define the free-particle routing time as the minimum time to implement a given permutation using an architecture-respecting free-particle Hamiltonian:
\begin{equation}
    \hrtf{G,p} \coloneqq \min_H [t \textrm{ s.t. } U(H,t) = U_p],
\end{equation}
where $U(H,t)$ is the time-evolution operator by $H$ for time $t$.
We define the \emph{worst-case free-particle routing time} for a graph $G$ as
\begin{equation}
    \hrtf{G} \coloneqq \max_{p \in \mathcal{S}_{N}} \rtf{G,p}.
\end{equation}
Since the dynamics of the creation and annihilation operators are identical for bosons and fermions, the routing time does not depend on the particle type.

One can also define a gate-based routing model for free particles. Allowed two-mode gates between mode $i$ and mode $j$ for $(i,j) \in E$ are the unitaries generated by $h_i c_i^\dagger c_i + h_j c_j^\dagger c_j + (h_{ij} c^\dagger_i c_j + \hc)$ for time-dependent $h_i$, $h_j$, and $h_{ij}$. As in the case of qubits, gates on disjoint pairs of vertices may be applied simultaneously, and application of a single layer of allowed simultaneous gates takes depth (or time) 1. 
We define the gate-based free-particle routing time $\rtf{G,p}$ as the minimum depth to implement a given permutation $p$ using an architecture-respecting free-particle circuit (i.e., a circuit composed of free-particle gates).
We define the \emph{worst-case} gate-based free-particle routing time for a graph $G$ as $\rtf{G} \coloneqq \max_{p \in \mathcal{S}_N} \rtf{G,p}$.

\subsection{Transpositions}
\begin{definition}[Graph Tripartition]
    A tripartition of a graph $G$ is a partitioning of the vertices of $G$ into three sets $L, C, R$ with $N_L, N_C, N_R$ vertices, respectively, such that there are no edges connecting vertices in $L$ to vertices in $R$. Without loss of generality, $L$ is taken to be the larger of $L,R$, i.e., $N_L \geq N_R$.
\end{definition}
A graph tripartition is depicted in \cref{fig:tripartite}.
\begin{definition}[Vertex bottleneck]
    In a graph with a tripartition, $C$ is a vertex bottleneck if $N_C \leq N_R$. 
\end{definition}
Throughout this paper, we consider tripartitions with a vertex bottleneck, i.e., we assume everywhere that $N_L \geq N_R \geq N_C.$
We focus on permutations that perform the maximum possible number $N_R$ of pairs of disjoint swaps between sites through a general vertex bottleneck (\cref{fig:tripartite}). On the star graph $S_{N+1}$, such a permutation can be obtained by dividing the leaves into two sets $L,R$ of $N/2$ qubits each, and swapping each qubit of $L$ with a corresponding qubit of $R$. We call this a global transposition, and it is depicted in \cref{fig:star_transposition}.   
\begin{figure}[t]
    \centering
\includegraphics[width=0.6\columnwidth]{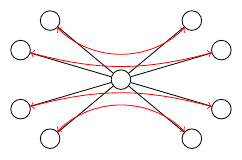}
\caption{A global transposition permutation on the star graph. The black lines denote the connectivity constraints. The red lines with arrows indicate pairs of qubits to be swapped. }
\label{fig:star_transposition}
\end{figure}

Our investigation of such permutations is motivated by the fact that they can be used at most twice 
to implement an arbitrary permutation.
Any cycle, or set of non-overlapping cycles, can be performed using two global transpositions. Since any permutation can be decomposed into a product of non-overlapping cyclic permutations, arbitrary permutations can be performed using two global transpositions. For example, the cyclic permutation $1\rightarrow 2,2\rightarrow 3,\ldots, 2n \rightarrow 1$, which we denote by $(1\; 2\ldots 2n)$, can be decomposed as a product of two global transpositions: 
\begin{align}
    (1\; 2\ldots 2n)&=(2\; 2n)(3\; 2n-1)\cdots (n\; n+2)\nonumber\\
    &\quad \times (1\; 2n)(2\;2n-1)\cdots(n\; n+1),
\end{align}
where the two lines on the right-hand side each consist of non-overlapping pairwise transpositions.
This is illustrated in \cref{fig:transposition}.
On the star graph, any permutation can be decomposed into a permutation that acts only on the leaves, followed by a single final swap with the center node. Hence, the time taken for arbitrary permutations on the star graph is predominantly determined by the time to do a global transposition.
\begin{figure}[t!]
    \centering
\includegraphics[width=0.6\columnwidth]{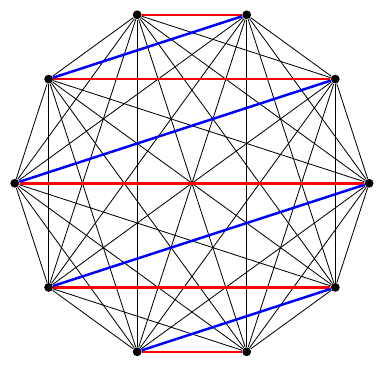}
\caption{
Any cyclic permutation can be decomposed into two stages of transpositions. For example, a clockwise cycle of all vertices can be achieved by a global transposition between pairs of sites connected by blue edges followed by a global transposition for red edges. To perform an anti-clockwise cycle, simply transpose along the red edges and then the blue edges.}
    \label{fig:transposition}
\end{figure}

\subsection{Routing and entanglement distribution}\label{sec:routing_entanglement}
Routing can be used to distribute entanglement, as depicted in \cref{fig:routing_entanglement}. This fact allows us to lower bound the routing time using bounds on the time to generate entanglement. We make use of two such bounds. 
\begin{figure}
    \subfloat[Initial state with local bell pairs]{\label{fig:entanglement_same} \includegraphics[width=0.9\columnwidth]{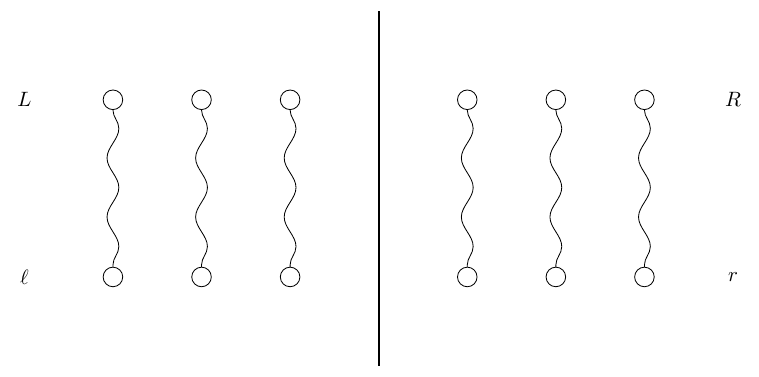}}
    \hfill 
    \subfloat[Bell pairs distributed after routing]{\label{fig:entanglement_opposite} \includegraphics[width=0.9\columnwidth]{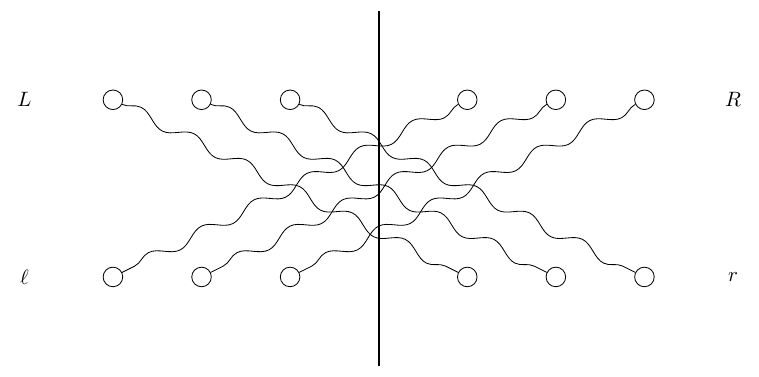}}
    \caption{Routing enables entanglement distribution. (a) The system has two parts $L, R$ that are initially entangled with hidden auxiliary systems $l, r$, respectively. (b) On swapping every qubit in $L$ with the corresponding qubit in $R$, the $L$ system becomes entangled with $r$, and $R$ becomes entangled with $l$, so we have generated 6 ebits of entanglement across the cut dividing $L$ and $R$.}
    \label{fig:routing_entanglement}
\end{figure}
The small incremental entangling theorem, whose conjecture is attributed to Kitaev in Ref.~\cite{bravyi_upper_bounds} and which is proved in Ref.~\cite{SIE}, bounds the rate at which entanglement can be generated between parts of a system in terms of the local dimension and the operator norm of the Hamiltonian coupling the parts. We restate this result here:
\begin{lemma}[Small Incremental Entangling (SIE)]
    \label{lem:SIE}
    Given a Hamiltonian $H = H_{Aa} + H_{Bb} + H_{AB}$ that acts on a system consisting of subsystems $a, A, B, b$ (\cref{fig:sie}), for any state $\rho$,
    \begin{align}
        \frac{\dd S_{aA}(\rho)}{\dd t} \le c \|H_{AB}\|\log(d), 
    \end{align}
    where $S_{Aa}(\rho) = -\trace{\rho_A \log \rho_A}$ is the entanglement entropy of $\rho$ across the $aA,Bb$ bipartition, $c$ is some positive constant ($c=2$ in \cite{Audenaert_2014}), and $d = \min\{\dim (A), \dim (B)\}$.
\end{lemma}
\begin{figure}[t!]
    \centering
\includegraphics[width=0.9\columnwidth]{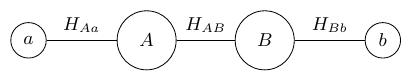}
\caption{Partitioning of system for SIE.}
    \label{fig:sie}
\end{figure}
We also make use of the small total entangling property \cite{STE, adv_lim}, which bounds the total amount of entanglement between parts of a system that can be generated by any unitary evolution, in terms of the local dimension of the parts.
\begin{lemma}[Small Total Entangling (STE) (Proposition 2 of Ref.~\cite{STE})]\label{lem:STE}
    In a system consisting of subsystems $a, A, B, b$, for any unitary acting only on $A$ and $B$, the maximum change in the entanglement entropy of the $Aa$ subsystem is bounded as
    \begin{equation}
        \Delta S_{Aa}(\rho) \leq 2 \log(\min\{\dim (A), \dim (B)\}).
    \end{equation}
\end{lemma}

In systems of free particles, we quantify bipartite entanglement by the \emph{mode entanglement} of the underlying Fock state \cite{Shi_2003, Zanardi_2002}, which is the von Neumann entropy obtained by tracing out a subset of the modes.

\begin{definition}[Mode entanglement]
\label{def:mode_entanglement}
    For a state of identical particles $\rho$ (i.e., a density matrix in the occupation number basis) with modes partitioned into $X$ and $\bar X$, the mode entanglement is
    \begin{equation}
        S_X(\rho) = -\trace{\rho_X \log \rho_X}.
    \end{equation}
\end{definition}

\section{Routing free particles} \label{sec:free_particles}

\subsection{An optimal protocol for routing free particles} \label{sec:free_protocol}
Here we consider the problem of routing free particles on a star graph. 
We give a protocol that can route $N$ free fermions on the star graph in time $\sqrt{N}$. Like the $W$-state-based routing protocol on the vertex barbell graph (\cref{fig:vertex_barbell}) of~\textcite{adv_lim}, this makes use of states spread over multiple sites (here, the leaves) and transfers them one-by-one through the central vertex.

\begin{theorem}
    $\hrtf{S_N} = \bigo{\sqrt{N}}$.
\end{theorem}
\begin{proof}
We consider the star graph $S_{2N+1}$, i.e., a star graph with an even number of leaves. To see that this is sufficient, consider any permutation $p$ on a star graph with an odd number of leaves. We can select some leaf $i$, and perform a modified permutation $p'$ which is identical to $p$ except when acting on sites $i$ and $p^{-1}(i)$. The permutation $p'$ instead maps $p^{-1}(i)$ to $p(i)$ and leaves $i$ unchanged, and is therefore a permutation on an even number of leaves. By first performing $p'$, and then swapping qubits $i$ and $p(i)$, we can perform the permutation $p$. Hence it suffices to consider a star graph with an odd number of vertices (i.e., an even number of leaves), since for any permutation $p$ on a star graph with an odd number of leaves $2N+1$, the routing time is only a constant (3) greater than the routing time for a corresponding permutation on a subgraph with an even number of leaves $2 N$.

As explained in \cref{sec:prelims}, an arbitrary permutation can be implemented by application of at most two global transpositions. For any global transposition, the leaf vertices can be divided into two sets $L$ (left) and $R$ (right) with $N$ leaves on each side, such that the transposition swaps each qubit in $L$ with the corresponding qubit in $R$. Let the pairs of creation and annihilation operators on each left (right) leaf $i$ be $l_i^\dagger$ and $l_i$ ($r_i^\dagger$ and $r_i$), respectively. On the center, we denote them as $c^\dagger$ and $c$.

The left Fourier modes are $f_{l,0}, \dots, f_{l,N-1}$, where 
\begin{equation} 
f_{l,k} \coloneqq \frac{1}{\sqrt{N}} \sum_j e^{-2\pi\ii kj/N} l_j.
\end{equation}
Likewise, we denote the right Fourier modes as $f_{r,0}, \dots, f_{r,N-1}$.

We divide our protocol into $N$ time steps, numbered $0$ to $N-1$.
At the $k$th time step, we perform the following operations:
\begin{itemize}
    \item  First, turn on coupling $H_{L,k} = \sum_j e^{-2\pi\ii j k} c^\dagger l_j + \hc$ for time $\frac{\pi}{2\sqrt{N}}$. This Hamiltonian can be equivalently written as $\sqrt{N} (c^\dagger f_{l,k} + \hc)$ This swaps the center mode with the $k$th Fourier mode on the left. This is illustrated in \cref{fig:Fourier_left}.
    \item  Repeat the same step but with the right leaves using the coupling $H_{R,k} = \sum_j e^{-2\pi\ii j k} c^\dagger r_j + \hc$ This interaction swaps $f_{l,k}$ (which sits at the center following the first step) with $f_{r,k}$ such that $f_{l,k}$ is now on the \textit{right} leaves, $f_{r,k}$ is on the center, and $c$ has been moved to the $k$th left Fourier mode. This is illustrated in \cref{fig:Fourier_right}.
    \item Finally, perform the coupling $H_{L,k} = \sum_j e^{-2\pi\ii j k} c^\dagger l_j + \hc$ with the left leaves again to swap $c$ and $f_{r,k}$. We have now effectively swapped Fourier modes $f_{l,k}$ and $f_{r,k}$, with $c$ remaining on the center. This is illustrated in \cref{fig:Fourier_left}.
\end{itemize}
\begin{figure}
    \subfloat[Left swap with center]{\label{fig:Fourier_left} \includegraphics[width=0.9\columnwidth]{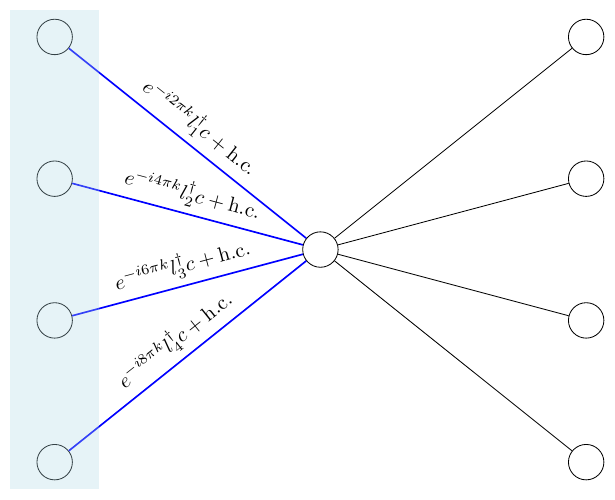}}
    \hfill 
    \subfloat[Right swap with center]{\label{fig:Fourier_right} \includegraphics[width=0.9\columnwidth]{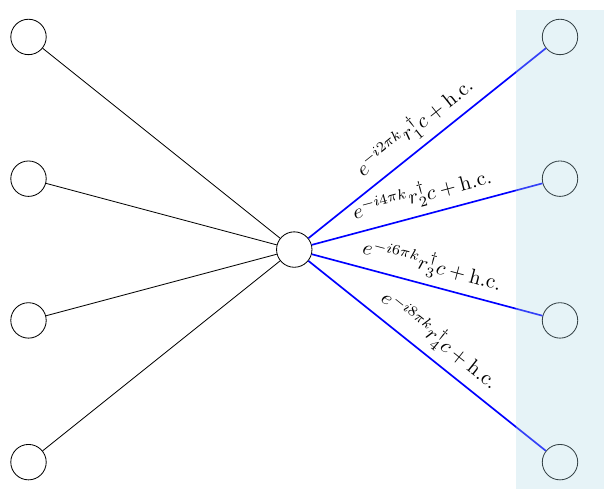}}
    \caption{Swapping the $k^{th}$ Fourier modes on the left and right sets of leaves.}
\end{figure}
Since both Fourier modes and spatial modes form orthonormal bases, swapping all of the corresponding Fourier modes between the left and the right vertices is equivalent to swapping each of the corresponding spatial modes. Since each step of this protocol swapping one Fourier mode at a time takes time $\bigo{\frac{1}{\sqrt{N}}}$, the whole protocol takes time $\bigo{\sqrt{N}}$.
As discussed in \cref{sec:prelims}, the dynamics of free bosonic and free fermionic creation and annihilation operators are identical. Hence this protocol performs routing for both bosons and fermions.
\end{proof}
 
\subsection{Routing qubits}
\label{sec:qubit_nogo}
 A natural question is whether a similar protocol exists for qubits. Unfortunately, simply converting the protocol above to a qubit protocol does not work, as we now explain.
 
 First, consider the fermionic version of this protocol. Applying the Jordan-Wigner \cite{Jordan_Wigner_1928} transform to convert the fermionic operators to qubit operators yields non-local terms, so the generated qubit Hamiltonian does not respect the architecture's geometry and hence does not produce a valid protocol. Even if we permitted ancillary qubits, using an approach such as the Verstraete-Cirac encoding \cite{Verstraete_Cirac_2005} to ensure that we respected locality would introduce a large number of ancillas ($\bigtheta{N}$) on the central qubit due to the large number of non-commuting terms in the Hamiltonian.
 
 Now we consider whether a bosonic protocol can be converted to a qubit protocol. One approach is to consider the bosonic protocol restricted to an initial state where each mode has occupancy at most 1, and replace the bosonic creation and annihilation operators with spin-1/2 raising and lowering operators, respectively. Another similar strategy to convert a bosonic protocol to a qubit routing protocol is by truncating the bosonic operators to a low  (constant) occupancy subspace by replacing the bosonic operators with qudit operators, which may yield an approximate protocol for qubits with ancillas. 
 
 Unfortunately, such an approach is infeasible.
 Consider a bosonic system starting in the uniform superposition of all bitstrings of length $N$ with Hamming weight $N/2$, i.e., the weight $N/2$ Dicke state $\ket{W^N_{N/2}}$.
Since every mode has occupancy at most 1, this is a valid initial state for a qubit system as well.
As before, we denote the left Fourier modes as $f_{l,0}, \dots, f_{l,N-1}$, where 
\begin{equation} 
f_{l,k} \coloneqq \frac{1}{\sqrt{N}} \sum_j e^{-2\pi\ii kj/N} l_j.
\end{equation}
Likewise, we denote the right Fourier modes as $f_{r,0}, \dots, f_{r,N-1}$.
Then, a simple calculation shows that the expected occupancy of the zeroth Fourier mode $f_{l,0}$ is $\bra{W_{N/2}^N}  f_{l,0}^\dagger f_{l,0} \ket{W_{N/2}^N} = \frac{\sqrt{N+2}}{4}$. Hence, if we were to truncate the protocol using qudit operators, it would require qudits of local dimension $\bigomega{\sqrt{N}}$. Such an approach therefore cannot yield a qubit routing protocol.

Though we have ruled out some obvious approaches, it may still be possible to produce a similar protocol for qubits (or to rule one out entirely), which we believe are interesting directions for future work.

\subsection{Lower bound on the free-particle routing time}
\label{sec:free_bound}

In this section, we give lower bounds on the time to route $N$ non-interacting particles (fermions or bosons) through a vertex bottleneck.

First, we show a lower bound of $ \bigomega{N_R/N_C}$ on the gate-based routing model for graphs with a vertex bottleneck. On the star graph, this bound yields $\rtf{S_N} = \bigomega{N}$. This implies that our Hamiltonian-based protocol in \cref{sec:free_protocol} obtains a quadratic speedup over gate-based routing.

Then, we show a lower bound on the Hamiltonian routing time $\hrtf{G}$ for any graph $G$ with a vertex bottleneck.
In particular, our proof implies that the free-particle routing time on the star graph is lower-bounded by $\rtf{S_N} =\bigomega{\sqrt{N}}$, which implies that the free-particle routing protocol given in \cref{sec:free_protocol} is optimal.

To prove these routing-time lower bounds, we first lower bound the time to generate $\bigtheta{N}$ mode entanglement (\cref{def:mode_entanglement}) between two halves of a free-particle system.
Analogously to the discussion in \cref{sec:routing_entanglement}, the ability of a system to route $N$ non-interacting particles across a partition in time $t$ can be reduced to the problem of generating $\bigtheta{N}$ mode entanglement in that time---for example, when $N$ modes, each of which is entangled with an ancillary mode as the state $\frac{1}{\sqrt{2}} (\ket{01} + \ket{10})$, are routed across a partition, the amount of entanglement generated across the partition is $\bigtheta{N}$. 
As such, a lower bound on the time to generate $\bigtheta{N}$ entanglement across a vertex bottleneck can be used to lower bound the time to perform routing.

To fix the model more precisely, we assume the system is defined on an architecture graph with a tripartition $V = L \cup C\cup R$ (cf.\ \cref{fig:tripartite}), where each vertex represents a mode and each edge indicates pairs of modes between which we are allowed to apply a gate.
We now show a lower bound on the gate-based routing time for free particles.
\begin{theorem}
    For any graph $G$ with a vertex bottleneck, $\rtf{G} = \bigomega{N_R/N_C}$.
\end{theorem}
\begin{proof}
We first consider routing in a system of free fermions. As discussed in \cref{sec:free_particles}, due to their identical operator dynamics, any lower bound for routing free fermions also applies to routing free bosons.
Consider an initial state $\rho_{\text{in}} = \frac{1}{2^{N_R}} \mathbb{I}_R \otimes \ket{0}^{N_C+N_L}$. This can be considered as a state where each mode in $R$ is maximally entangled with some reference system that has been traced out. The initial von Neumann entropy of the subsystem $R$ is thus $S_R = N_R$. On routing every mode from $R$ to $L$, all the entanglement (or entropy) is transferred to $L$, and $S_R = 0$.
Following a similar argument to one presented in~\textcite{adv_lim}, by a straightforward application of STE (\cref{lem:STE}) to free fermions, the maximum change in the von Neumann entropy of $R$ when applying a circuit of depth $d$ is
\begin{equation}
    \Delta S_R \leq 2dN_C.
\end{equation}
On the other hand, to route every mode in $R$ to $L$ requires a total change in the von Neumann entropy of $R$ of $|\Delta S_R| = N_R$.
Hence, the minimum circuit depth to perform routing is $\rtf{G} = \bigomega{N_R/N_C}$.
\end{proof}
Applied to the star graph, this implies a bound of $\rtf{S_N} = \bigomega{N}$, and hence our Hamiltonian routing protocol from \cref{sec:free_protocol} obtains a quadratic speedup over gate-based routing.

We now show a lower bound on Hamiltonian free-particle routing in graphs with a vertex bottleneck. As before, each vertex represents a mode and each edge corresponds to a possible two-body ``hopping'' term in the Hamiltonian. 
Let the particles evolve under the following non-interacting (time-dependent) Hamiltonian:
\begin{align}
    \label{eq:freeHam}
 	H(t) &= H^{(LL)} + H^{(CC)} + H^{(LC)} \nonumber\\
  &= \sum_{i, j\in L}H^{(LL)}_{ij}(t)l_i^\dagger l_j + \sum_{i,j\in C}H^{(CC)}_{ij}(t)c_i^\dagger c_j \nonumber\\
  &+ \sum_{\substack{i \in L \\ j \in C}} \left[H^{(LC)}_{ij}(t) l_i^\dagger c_j + \hc\right] \nonumber \\
  &+ \sum_{\substack{i \in R \\ j \in C}} \left[H^{(RC)}_{ij}(t) r_i^\dagger c_j + \hc\right] + \sum_{i, j\in R}H^{(RR)}_{ij}(t)r_i^\dagger r_j,
\end{align}

where the $l_i^\dagger$ ($l_i$) are the creation (annihilation) operators on $L$, $c_i^\dag$ ($c_i$) are the corresponding operators on $C$, and $r_i^\dag$ ($r_i$) are the corresponding operators on $R$.

As before, we assume that Hamiltonian terms coupling $L$ and $C$ are bounded by $\norm{H^{(LC)}_{ij}(t)} \le 1$, for all times $t$. 

The following lemma bounds the entangling rate in a free fermionic system. We use this to bound the free fermionic Hamiltonian routing time, and therefore the Hamiltonian routing time for free bosons as well.

\begin{lemma}
    \label{lem:free_entangling}
For a free fermionic system defined on a graph with a vertex bottleneck evolving under the free-particle Hamiltonian defined in \cref{eq:freeHam}, the entangling rate between subsystems $L$ and $C \cup R$ is bounded as
\begin{equation}
    \frac{\dd S_L(\rho(t))}{\dd t}\le 4c N_C \sqrt{N_L}
\end{equation} 
\end{lemma}
\begin{proof}
The proof of the theorem follows from applying the Small Incremental Entangling (SIE) theorem (\cref{lem:SIE}) to a free fermionic Hamiltonian acting on a bipartition of the system $(L, C \cup R)$. 

By SIE, we are only concerned with $H^{(LC)} = \sum_{\substack{i \in L \\ j \in C}} \left[H^{(LC)}_{ij}(t) l_i^\dagger c_j + \hc\right]$.
Consider the single normalized mode
\begin{equation}
\tilde l_j^\dag(t) \coloneqq \frac{1}{\sqrt{\sum_{i \in L} \left|H^{(LC)}_{ij}(t)\right|^2}}\sum_{i \in L} H^{(LC)}_{ij}(t) l^\dag_i.
\end{equation}
Then
\begin{equation}
    H^{(LC)} = \sum_{j \in C} h_j(t)
\end{equation}
where $h_j(t) = \sqrt{\sum_{i \in L} \left|H^{(LC)}_{ij}(t)\right|^2} \left[\tilde l_j^\dag(t) c_j + \hc\right]$.

Taking the same approach as in the proof of the area law for quasi-adiabatic continuation in \cite{SIE}, we can bound the entangling rate as
\begin{align}
\begin{split}
    \label{eq:SIE_fat_center}
    \frac{\dd S_L(\rho(t))}{\dd t}  &= -\ii \trace{H^{(LC)}(t)[\rho(t),\log(\rho_L(t) \otimes \mathbf{1}_C)]}\\
    &= -\ii\sum_{j \in C} \trace{ h_j(t)[\rho(t),\log(\rho_L(t) \otimes \mathbf{1}_C)]}\\
    &\le 2c\sum_{j\in C} \|h_j(t)\|,
\end{split}
\end{align}
where $\rho_L(t) = \ptrace{C}{\rho(t)}$, and $c$ is the constant in \cref{lem:SIE}.
In the third line, we applied the SIE theorem to bound the contribution of each $h_j(t)$ term separately (cf.\ Eq. (18) in \cite{SIE}).

To finish the proof, we bound the norm of $h_j(t)$ as follows:

\begin{align}
\begin{split}
    \|h_j(t)\| &= \left\lVert\sqrt{\sum_{i \in L} \left|H^{(LC)}_{ij}(t)\right|^2} \left[\tilde l_j^\dag(t) c_j + \hc\right] \right\rVert \\
    &\le 2\sqrt{N_L}, \label{eq:HItbound}
\end{split}
\end{align}
where we have used the fact that, for fermions,  $\norm{\tilde l_j^\dag(t) c_j} \leq 1$.
Plugging this bound on $\|h_j(t)\|$ into \cref{eq:SIE_fat_center} yields
\begin{equation}
    \frac{\dd S_L(\rho(t))}{\dd t}\le 4c N_C \sqrt{N_L},
\end{equation}
as claimed.
\end{proof}

\begin{theorem}
    \label{thm:free_lower}
For a system defined on a graph with a vertex bottleneck evolving under the free-particle Hamiltonian defined in \cref{eq:freeHam}, 
\begin{equation}
    \hrtf{G} = \bigomega*{\frac{N_R}{N_C \sqrt{N_L}}}.
\end{equation}
\end{theorem}
\begin{proof}
   We start with the case of free fermions. By \cref{lem:free_entangling}, the entangling rate satisfies $\frac{\dd S_L(\rho(t))}{\dd t}\le 4c N_C \sqrt{N_L}$.
   Hence the time to generate $S_L(\rho(t)) = \bigtheta{N_R}$ ebits of entanglement between subsystems $L$ and $R \cup C$ is lower bounded by 
\begin{equation}
    t = \bigomega*{\frac{N_R}{\sqrt{N_L}N_C}}.
\end{equation} 
If every mode in $R$ is initially maximally entangled with an ancillary mode, then routing every mode from $R$ to $L$ requires producing $N_R$ bits of mode 
entanglement between $L$ and the rest of the system. Hence, the routing time for free fermions (and thus for free bosons as well) is lower bounded as $\bigomega*{\frac{N_R}{\sqrt{N_L}N_C}}$ as claimed.
\end{proof}

We now apply \cref{thm:free_lower} to the case of the star graph. 
Setting $N_L = N$ and $N_C = 1$ yields a lower bound of $t = \bigomega{\sqrt{N}}$ on the free-particle routing time, which matches the upper bound given by the free-fermion routing protocol in \cref{sec:free_protocol}. 
As such, for free particles, the optimal routing time on the star graph is $\hrtf{S_N} =\bigtheta{\sqrt{N}}$.

\section{Lower bounds on qubit routing through vertex bottlenecks}
\label{sec:qubit_bound}

We now return to the more common setting of qubits rather than free particles.
An architecture-respecting Hamiltonian for a graph with a tripartition can be decomposed as \begin{equation}\label{eq:H=LR}
    H = H_L + H_C + H_R,
\end{equation}
where $H_L$ consists of local terms in $L$, terms coupling different sites in $L$, and coupling sites in $L$ to sites in $C$; $H_C$ consists of terms local in $C$ or coupling different sites in $C$; and $H_R$ consists of terms local in $R$, coupling different sites in $R$, and coupling sites in $R$ to sites in $C$. These terms can be further expanded in terms of Pauli operators $\{X^0, X^1, X^2, X^3\}$: \begin{subequations}\label{eq:H=Pauli}
\begin{align}
    H_L &= \sum_{l_1}^{N_L} h_{l_1}^{(L)\alpha_1} X^{\alpha_1}_{l_1} + \sum_{l_1< l_2} ^{N_L, N_L}  h_{l_1, l_2}^{(LL)\alpha_1\alpha_2} X_{l_1}^{\alpha_1} X_{l_2}^{\alpha_2} \nonumber \\ &\qquad + \sum_{l_1, c_1} ^{N_L, N_C}  h_{l_1, c_1}^{(LC)\alpha_1\gamma_1} X_{l_1}^{\alpha_1} X_{c_1}^{\gamma_1}, \\
    H_C &= \sum_{c_1}^{N_C} h_{c_1}^{(C)\gamma_1} X^{\gamma_1}_{c_1} + \sum_{c_1< c_2} ^{N_C, N_C}  h_{c_1, c_2}^{(CC)\gamma_1\gamma_2} X_{c_1}^{\gamma_1} X_{c_2}^{\gamma_2}, \\
    H_R &= \sum_{r_1}^{N_R} h_{r_1}^{(R)\beta_1} X^{\beta_1}_{r_1} +\sum_{r_1< r_2} ^{N_R, N_R}  h_{r_1, r_2}^{(RR)\beta_1\beta_2} X_{r_1}^{\beta_1} X_{r_2}^{\beta_2} \nonumber \\ &\qquad+ \sum_{r_1, c_1} ^{N_R, N_C}  h_{r_1, c_1}^{(CR)\beta_1\gamma_1} X_{r_1}^{\beta_1} X_{c_1}^{\gamma_1},
\end{align}
\end{subequations}
where we use the Einstein summation convention (i.e.\ summing over repeated indices) for  $\alpha_1, \alpha_2,\beta_1, \beta_2, \gamma_1, \gamma_2 \in\{1,2,3\}$. We assume the coefficients satisfy
\begin{align}
    &\abs{h_{l_1}^{(L)\alpha_1}}, \abs{h_{c_1}^{(C)\gamma_1}}, \abs{h_{r_1}^{(R)\beta_1}} \le \sqrt{N}, \notag \\
    &\abs{h_{l_1, l_2}^{(LL)\alpha_1\alpha_2}},
    \abs{h_{l_1, c_1}^{(LC)\alpha_1\gamma_1}} \leq 1, \notag \\ &\abs{h_{c_1, c_2}^{(CC)\gamma_1\gamma_2}}, \abs{h_{r_1, r_2}^{(RR)\beta_1\beta_2}}, \abs{h_{r_1, c_1}^{(CR)\beta_1\gamma_1}} \le 1 \label{eq:LR<1}
\end{align}
with $N=N_L + N_R + N_C$ and $N_C \leq N_L, N_R$. It is natural to assume even the single-site terms must be bounded by 1. However, some routing models may allow for faster local fields \cite{adv_lim}. In this work, we assume a bound of $\sqrt{N}$ for technical reasons to retain our bound on the Frobenius commutator norm in \cref{lem:H<sqrtN_app}.
Since our graph has a bottleneck, $N_L \geq N_R \geq N_C$. These coefficients may be time-dependent such that $H$ is piecewise time-independent (recall \cref{def:piecewise}) with minimum segment width $\propto 1/\sqrt{N_L}$ (this value is due to the limitations of our techniques, as expressed by \cref{eq:segment_restriction}).
Our Hamiltonian can alternatively be written as
$H = H_{LC} + H_R$,
where $H_{LC} = H_L + H_C$.

Our main result is a lower bound on the Hamiltonian routing time for any graph with a vertex bottleneck.
\begin{theorem}
    \label{thm:qubit_lower}
    For any graph $G$ with a tripartition into $L,C,R$ with $N_L, N_C, N_R$ vertices, respectively, if $C$ is a vertex bottleneck, then for any constant $\delta \in (0, 1/3]$, if $N_R > 4(2\times 5^{\frac{1-\delta}{2\delta}})N_C + 2$, there is a constant $w_{\delta}$ such that  the routing time for piecewise time-independent Hamiltonians with minimum segment width $\Delta = w_{\delta}/\sqrt{N_L}$ is 
    \begin{equation}
        \hrtd{G} = \bigomega*{\frac{N_R^{1-\delta}}{\sqrt{N_L}N_C}}.
    \end{equation}
\end{theorem}

\textit{Proof sketch:}
Our proof consists of two parts. In the first part, in \cref{lem:circuit_approx}, we show that any evolution $U(H,t)$ by an architecture-respecting piecewise time-independent Hamiltonian $H$ with minimum segment width $w_{\delta}/\sqrt{N_L}$ for time $t\leq c_\delta N_R^{1-\delta}/(\sqrt{N_L}N_C)$ can be well-approximated by an architecture-respecting circuit $\tU$ of depth $d \leq N_R/(4N_C) + 2\times 5^{\frac{1-3\delta}{2\delta}}$ for some constant $c_\delta$. To show this, we first prove  \cref{lem:trotter_frobenius}, which bounds the approximation error of Trotter-Suzuki product-formula circuit approximations to $U(H,t)$ in systems with a tripartition. This lemma is based on the fact that vertex bottlenecks constrain the Frobenius norm of higher-order commutators of terms in any Hamiltonian that respects the tripartition connectivity, which we prove in \cref{lem:H<sqrtN_app}. We further make use of recent work \cite{Trotter_Frob22} that bounds the error of Trotter-Suzuki product-formula circuit approximations in terms of the higher-order-commutator Frobenius norm. In the second part of the proof, we show that such a circuit cannot approximate a desired permutation unitary well. According to STE (\cref{lem:STE}), the amount of entanglement between $L$ and $C\cup R$ that can be generated by an architecture-respecting circuit of depth $d$ is bounded by $2dN_C$.
There exist permutations $p$ such that the permutation unitary $U_p$ can be used to increase the entanglement between $L$ and $R$ by $N_R$,
an increase which, by STE, requires a circuit of depth $N_R/ (2N_C)$. In \cref{lem:circuit_entanglement}, we show that such a circuit must be far from the target unitary $U_p$. Now combining the two parts of the proof, since we know that any architecture-respecting evolution $U(H,t)$ for $t\leq  c_\delta N_R^{1-\delta}/\sqrt{N_L}N_C$ is well approximated by a circuit $\tU$ that cannot perform the permutation $p$, our main result of a lower bound on the Hamiltonian routing time follows. We defer the formal proof to the end of this section.

In \cref{sec:frob_proof}, we show the following bound on the error of approximating architecture-respecting evolutions by circuits obtained using the Trotter-Suzuki formula.
\begin{lemma}
\label{lem:trotter_frobenius}
Consider a graph with a vertex bottleneck. Let $H$ be any time-independent Hamiltonian that respects this connectivity.
    Let $\tU$ be the architecture-respecting simulation circuit corresponding to dividing the time-evolution $U=U(H,t)$ into $M$ equal segments, and simulating each by the $(2k)^{\mathrm{th}}$-order Trotter-Suzuki formula.
    Then, $\tU$ has depth $d = 2 \times 5^{k-1} M$ \cite{Trotter_Frob22}
    and there exists a function $g(k)$, only dependent on $k$, such that
    \begin{equation}
    \label{eq:U-tU<comm}
    \norm{U- \tU}_{\rm F} \le g(k) \frac{t^{2k+1}}{M^{2k}} \sqrt{{N_L}^{2k} N_R N_C}.
    \end{equation}
\end{lemma}
We use this lemma to show that the circuit obtained from the Trotter-Suzuki formula can be an arbitrarily good approximation of the time-evolution.
\begin{lemma}
    \label{lem:circuit_approx}
    Consider a graph with a vertex bottleneck. Let $H$ be any piecewise time-independent Hamiltonian with minimum segment width $\Delta$ that respects the connectivity of a graph with a tripartition.
    Let $U=U(H,t)$ be the unitary corresponding to evolution by $H$ for time $t$.
    Then, for any constant $\delta \in (0, 1/3]$ and $\epsilon > 0$,
    there exist constants $c_{\delta, \epsilon}, w_{\delta, \epsilon} > 0$ and $k \in \mathbb N_+$ such that, for \begin{equation}\label{eq:T<cdelta}
        t \leq  \frac{c_{\delta, \epsilon} N_R^{1-\delta}}{\sqrt{N_L}N_C}
    \end{equation}
    and $\Delta > w_{\delta, \epsilon}/\sqrt{N_L}$,
    there exists an architecture-respecting circuit $\tU$
    with depth $d \leq N_R/(4N_C) + 2 \times 5^{k-1}$ such that $\norm{\tU-U}_F \leq \epsilon$.

\end{lemma}
\begin{proof}

We first assume $H$ is time-independent for simplicity, and consider the piecewise case at the end of the proof.
Let $k$ be the smallest positive integer such that $\frac{1}{2k+1} \leq \delta$.

Selecting $M = \lceil \frac{1}{4}\left(2 \times 5^{k-1} \right)^{-1} N_R/N_C \rceil$, the circuit approximation $\tU$ obtained from \cref{lem:trotter_frobenius} has depth $d  = 2 \times 5^{k-1} M \leq N_R/(4N_C) + 2 \times 5^{k-1}$.
By \cref{lem:trotter_frobenius},
\begin{align}\label{eq:U-tU<tN}
    &\norm{U- \tU}_{\rm F} \spliteq \le 2^{6k} (5^{k-1})^{2k}  g(k) \spliteq \qquad\times t^{2k+1}\lr{\frac{N_C}{N_R}}^{2k} \sqrt{{N_L}^{2k} N_R N_C}.
\end{align}
Bounding the right-hand side by $\epsilon$ and solving for $t$, we obtain
\begin{equation}
    t \le \frac{\left(2^{6k}5^{2k^2-2k} g(k) \epsilon^{-1}\right)^{\frac{-1}{2k+1}} {N_R}^{1 - \frac{1}{2k+1}}}{\sqrt{N_L}N_C} \label{eq:T<ck}.
\end{equation}
This proves the Lemma for time-independent $H$, because there is a suitable constant $c_{\delta, \epsilon}$ such that \cref{eq:T<cdelta} implies \cref{eq:T<ck}, which guarantees $\norm{U- \tU}_{\rm F}\le \epsilon$.

The above results easily generalize to piecewise time-independent Hamiltonians with minimum segment width $w_{\delta, \epsilon}/\sqrt{N_L}$.
The Trotter error formula \cref{eq:U-tU<comm} does not hold directly because the Hamiltonian changes with time in some of the time windows of duration $t/M$.
Nevertheless, 
\begin{equation}
\label{eq:segment_restriction}
\frac{t}{M}\le \frac{(8\times 5^{k-1})c_{\delta,\epsilon}}{N_R^\delta \sqrt{N_L}} < \frac{w_{\delta, \epsilon}}{\sqrt{N_L}}
\end{equation}
for constant $w_{\delta, \epsilon}$ determined by $c_{\delta, \epsilon}$, which will be chosen shortly. Since $t/M$ is smaller than the minimum segment width $w_{\delta,\epsilon}/\sqrt{N_L}$, 
each time window of duration $t/M$ can be split into at most two smaller time windows,
each containing time-independent evolution.
Applying the Trotter error formula to these possibly smaller time windows,
we obtain \cref{eq:U-tU<comm} again but with an extra prefactor of $2$,
because there are at most $2M$ new time windows of duration at most $t/M$. The extra factor of $2$ propagates to \cref{eq:U-tU<tN}, so we only need to choose a slightly smaller $c_{\delta,\epsilon}$ than the time-independent case (say, half of its value) to guarantee that \cref{eq:T<cdelta} still implies $\norm{U- \tU}_{\rm F}\le \epsilon$. Hence $w_{\delta, \epsilon} = (4\times 5^{k-1})c_{\delta,\epsilon}$ suffices. 
\end{proof}

We now show that the approximation circuit $\tU$ of \cref{lem:circuit_approx} cannot achieve the permutation that routes the maximum number of qubits through the bottleneck. 
\begin{lemma}
\label{lem:circuit_entanglement}
	For any graph $G$ with a vertex bottleneck, consider a permutation $p$ that maps $m$ qubits from $R$ to $L$ through the vertex bottleneck
    and its associated permutation unitary $U_p$.
    Then, for any quantum circuit $\tilde{U}$ with depth $d<m/N_C$,
    \begin{equation}
        \norm{\tU-U_p}^2_{\rm F} > \frac{1}{4} \lr{\frac{m-2dN_C-1}{N_R}}^2.
    \end{equation}
\end{lemma}
\begin{proof}

First, we bound $\norm{\tU- U_p}_{\rm F}^2$ in terms of the trace distance. For any bit string $z$ of length $N_R$, we define 
\begin{equation}
    \rho_z:= \ket{z}_{R} \bra{z} \otimes (2^{-N_L-N_C}I)_{LC}.
\end{equation}
One can purify the system by a $2^{N_L + N_C}$-dimensional ancilla (labeled by $A$) so that
\begin{equation}
\label{eq:rho_purification}
    \rho_z = \ptrace{A}{\ket{\Psi(z)} \bra{\Psi(z)}}.
\end{equation}

Let $U' = U\otimes I_A$ and $\tU' = \tU\otimes I_A$. By definition of the Frobenius norm, and introducing $I = 2^{N_L + N_C}\sum_z \rho_z$,
\begin{align}
 &\norm{U - \tU}_{\rm F}^2 \nonumber \\
 &= \label{eq:frob_trace1} \frac{1}{2^{N_R}}\trace*{\sum_z \rho_z \abs{U-\tU}^2} \nonumber \\
 &= \frac{1}{2^{N_R}}\sum_z \brakett{\Psi(z)}{\abs{U-\tU}^2 \otimes I_A}{\Psi(z)} \nonumber\\
 &= \frac{1}{2^{N_R}} \sum_z \left[2 - 2 \Real{\brakett{\Psi(z)}{U'^\dagger \tU'}{\Psi(z)}} \right].
 \end{align}
 Using the fact that $\Real{x} \leq |x|$, and for $x \in [0,1], 1-x^2 \leq 2(1-x)$, we obtain
 \begin{align}
 &\sum_z \left[ 2 - 2 \Real{\brakett{\Psi(z)}{U'^\dagger \tU'}{\Psi(z)}} \right] \nonumber \\
 &\ge \sum_z 2(1-\abs{\brakett{\Psi(z)}{U'^\dagger \tU'}{\Psi(z)}}) \nonumber\\
 &\ge \sum_z 1-\abs{\brakett{\Psi(z)}{U'^\dagger \tU'}{\Psi(z)}}^2 \nonumber\\
 &=\label{eq:fidelity_trace_norm}\frac{1}{4} \sum_z \norm{U'\ketbra{\Psi(z)}{\Psi(z)}U'^\dagger - \tU' \ketbra{\Psi(z)}{\Psi(z)} \tU'^\dagger}_1^2 \\
 & \geq \frac{1}{4} \sum_z \norm{\traceOp_A\big[U'\ketbra{\Psi(z)}{\Psi(z)}U'^\dagger \nonumber \\ &\qquad\qquad\qquad\qquad\qquad - \tU' \ketbra{\Psi(z)}{\Psi(z)} \tU'^\dagger \big]}_1^2 \nonumber\\
&= \label{eq:frob_trace2}\frac{1}{4} \sum_z \norm{\tU\rho_z \tU^\dagger-U_{\rm p}\rho_z U_{\rm p}^\dagger}_1^2,
\end{align}
where in \cref{eq:fidelity_trace_norm}, we have used the relation between fidelity and trace-norm for pure states.

By the small total entangling theorem (\cref{lem:STE}), $\tU$ with depth $d$ can increase the entanglement across the $LC, R$ bipartition by at most $2d N_C$, i.e.,
\begin{equation}\label{eq:STE}
    S_R(\tU\rho \tU^\dagger) - S_R(\rho) \le 2d N_C,
\end{equation}
where $S_R$ is the von Neumann entropy of the reduced density matrix on the right $N_R$ qubits, and the vertex boundary of these qubits contains $N_C$ qubits as in \cref{fig:tripartite}.

Observe that $\rho_z$ is pure on $R$: $S_R(\rho_z)=0$, so \cref{eq:STE} implies \begin{equation}\label{eq:SL<N2}
    S_R(\tU\rho_z \tU^\dagger) \le 2d N_C.
\end{equation}
On the other hand, $U_p$ routes $m$ identity operators in $LC$ to $R$, so that \begin{equation}\label{eq:SL=N}
    S_R(U_{\rm p}\rho_z U_{\rm p}^\dagger) =m.
\end{equation}
Intuitively, two unitaries that produce different changes in the von Neumann entropy of the same initial state must be distant. By the Fannes-Audenaert Inequality~\cite{Audenaert_2007} (the second inequality below), 
\begin{align}\label{eq:tUrho-P>}
    &\norm{\tU\rho_z \tU^\dagger-U_{\rm p}\rho_z U_{\rm p}^\dagger}_1 \spliteq
    \ge \norm{ \ptrace{LC}{\tU\rho_z \tU^\dagger-U_{\rm p}\rho_z U_{\rm p}^\dagger } }_1 \spliteq
    \ge \frac{S_R(U_{\rm p}\rho_z U_{\rm p}^\dagger)-S_R(\tU\rho_z \tU^\dagger)-1}{\log \mathrm{dim} \mathcal{H}_R} \spliteq
   \ge \frac{m-2dN_C-1}{N_R}.
\end{align}
Using this to lower bound \cref{eq:frob_trace2} gives the result.
\end{proof}

Combining \cref{lem:circuit_approx} and \cref{lem:circuit_entanglement}, we may now prove our main result.
\begin{proof}[Proof of \cref{thm:qubit_lower}]
Let $p$ be a permutation that routes all $N_R$ qubits from $R$ to $LC$. 
Let $U$ be the evolution by any architecture-respecting piecewise time-independent Hamiltonian $H$ with minimum segment width $w_{\delta, \epsilon}/\sqrt{N}$ for time $t \leq c_{\delta, \epsilon} N_R^{1-\delta}/\sqrt{N_L}N_C$ for any $\delta>0$. By \cref{lem:circuit_approx}, $U$ can be $\epsilon$-approximated by an architecture-respecting circuit $\tU$ of depth at most $ N_R/(4N_C) + 2 \times 5^{k-1}$ such that $\norm{\tU-U}_{\rm F} \leq \epsilon$.  
Now, \cref{lem:circuit_entanglement} implies
\begin{align}
    \norm{\tU-U_{\rm p}}_{\rm F}^2 &\geq \frac{1}{4} \lr{\frac{N_R-2dN_C-1}{N_R}}^2 \nonumber\\&\geq \frac{1}{4} \lr{\frac{1}{2} - \frac{2(2\times 5^{k-1})N_C}{N_R}-\frac{1}{N_R}}^2. 
\end{align}
We thus have
\begin{align}
     \norm{\tU-U_{\rm p}}_{\rm F} &\geq \frac{1}{4} - \frac{(2\times 5^{k-1})N_C}{N_R} - \frac{1}{2N_R} \nonumber\\
     &\geq \frac{1}{4} - \frac{2(2\times 5^{k-1})N_C+1}{2N_R} .
\end{align}
    From the reverse triangle inequality, we have
\begin{align}
    \norm{U-U_{\rm p}}_{\rm F} &\geq \left|\norm{\tU-U_{\rm p}}_{\rm F} - \norm{\tU-U}_{\rm F} \right| \nonumber\\
    &\geq \frac{1}{4} - \frac{2(2\times 5^{k-1})N_C+1}{2N_R} - \epsilon.
\end{align}
Since $k$ is the smallest positive integer such that $\frac{1}{2k+1} \leq \delta$,
\begin{align}
    &\frac{1}{2(k-1)+1} > \delta \nonumber\\ &\implies k-1 < \frac{1-\delta}{2\delta}.
\end{align}
Therefore,
\begin{align}
    & N_R > \frac{4(2\times 5^{\frac{1-\delta}{2\delta}})N_C + 2}{1-8\epsilon} \nonumber\\
     \implies & N_R > \frac{4(2\times 5^{k-1})N_C + 2}{1-8\epsilon} \nonumber\\
     \implies & \norm{U-U_{\rm p}}_{\rm F} > \epsilon,
\end{align}
and $U$ cannot achieve the desired permutation.
Selecting any constant $\epsilon$ allows us to pick $w_{\delta} = w_{\delta, \epsilon}$ for which the theorem holds.
\end{proof}

A straightforward application of \cref{thm:qubit_lower} to the star graph yields the following result.
\begin{corollary}
For any $\delta \in (0, 1/3]$, if $N> 8(2\times 5^{\frac{1-\delta}{2\delta}}) + 4$, 
the Hamiltonian routing time for the star graph on $N+1$ vertices with $\Delta = \frac{1}{\sqrt{N}}$ satisfies $\hrtd{S_{N+1}} = \bigomega{N^{\frac{1}{2}-\delta}}$.
\end{corollary}

\section{Entanglement flow through bottlenecks}\label{sec:ent_cap}

In this section, we investigate the time taken to generate entanglement between subsystems connected by a bottleneck. In a system with a vertex bottleneck, the small total entangling theorem (\cref{lem:STE}) bounds the amount of entanglement that can be generated between $L$ and $RC$ by a circuit of depth $t$ as $\Delta S_L(\rho) \leq 2tN_C$. On the other hand, the small incremental entangling theorem (\cref{lem:SIE}) \cite{SIE} and other results on entanglement flow \cite{flow} bound the entangling rate in terms of the Hamiltonian operator norm as $\frac{\dd S_L(\rho)}{\dd t} \le c \|H_{LC}\|N_C $. When $N_C \ll N_L$, these two bounds differ greatly. Integrating the bound from \cref{lem:SIE} for an evolution of time $t$, we obtain 
\begin{equation}
    \Delta(S_L(\rho)) \le c N_LN_Ct.
\end{equation}  
Therefore, if this bound were saturable, by evolving for the same amount of time in a Hamiltonian-model rather than a gate-based model, we would obtain a speedup by a factor of $N_L$ in our ability to generate entanglement.

At first, this seems very counter-intuitive. Indeed, in systems of free fermions, we showed a tighter bound of $\bigo{N_C\sqrt{N_L}}$ on the entangling rate between $L$ and $RC$ in \cref{lem:free_entangling}.
On the other hand, in systems of free bosons, the entangling rate can be very large, depending on the initial state. Consider a system with just two modes.
If the initial state is the Fock state $\ket{N, 0}$, a simple calculation shows that applying a 50/50 beam splitter on the two modes leads to a mode entanglement of $ \bigtheta{\log N}$.

A simple example illustrates that a large increase of the entanglement rate with continuous evolution may be possible for qubits as well. Consider the system depicted in \cref{fig:ghz}. In a star graph, starting with an initial GHZ state on the leaves and $\ket{0}$ on the center, we can evolve into a GHZ state on the leaves and the center together in time $\pi/N$. Since the GHZ state has a von Neumann entropy of 1 between a single qubit and the remaining qubits, this implies that the entangling rate at some time must have been at least linear in $N$.

\begin{figure}
    \includegraphics[width=0.5\linewidth]{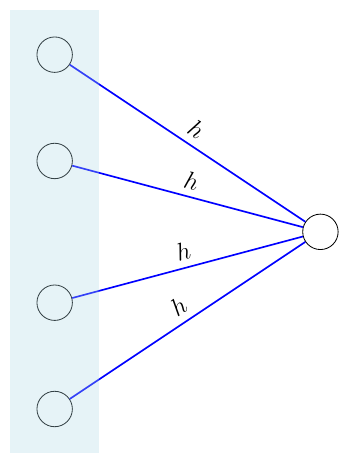}
    \caption{Linear entangling rate between $N$ leaves of the star graph and the center vertex. If the initial state is a product state of the $N$-qubit GHZ state $\frac{1}{\sqrt{2}}(\ket{00\dots0} + \ket{11\dots1})$ on the leaves, and $\ket{0}$ on the center, then evolving for time $\pi/N$ with interaction $h=\ketbra{1-}{1-}_{l_i, c}$ (i.e., the generator of a CNOT) on each edge results in a final $(N+1)$-qubit GHZ state on all the vertices. 
    }
    \label{fig:ghz}
\end{figure}

Although this system illustrates a worst-case counterexample for the instantaneous entangling rate through a bottleneck, the possibility of the bottleneck restricting our ability to produce entanglement between $L$ and $R$ remains open. Indeed, if it were always possible to generate fast entanglement, starting in any state, we could use this as a resource to perform routing. The existence of states that allow for fast entangling rates was used in Ref.~\cite{adv_lim} to obtain a fast routing protocol for the vertex barbell graph, although it also required very fast interactions with ancilla qubits. Our lower bound on the routing time in systems with a vertex bottleneck gives a hint that such states may be rare. Following this hint, we instead look at the entanglement that can be generated \textit{on average} starting with a state drawn from a 1-design.

We define the entanglement capacity $C_{\Delta}(\ket{\psi}, t):=|S_L(U(H,t)\ketbra{\psi}{\psi} U(H,t)^\dagger - S_L(\ketbra{\psi}{\psi})|$ as the amount of entanglement between $L$ and $RC$ that can be generated by an architecture-respecting piecewise time-dependent Hamiltonian $H$ of minimum segment width $\Delta$ in time $t$, starting in the state $\ket{\psi}$. 
Given a pure-state ensemble $\mu$, we define the average entangling capacity as $\Cavg$.
We prove that, in a system with a vertex bottleneck, the average entangling capacity for any 1-design \cite{Ambainis_Emerson_2007,Gross_Audenaert_Eisert_2007,Renes_Blume-Kohout_Scott_Caves_2004} scales with $N_L$ as $\bigo{\sqrt{N_L}}$.
This can be contrasted with the worst-case scaling in $N_L$ of $\bigo{N_L}$ for the entanglement capacity that is obtained by integrating the entangling rate obtained from \cref{lem:SIE}.

\begin{theorem}\label{thm:entangling_capacity}
    Let $\mu$ be a pure-state 1-design and let $H$ be a piecewise time-dependent Hamiltonian with minimum segment width $\Delta = 1/\sqrt{N_L}$ that respects the connectivity constraints of a graph with a vertex bottleneck.
    Then for any $0 < \delta \leq 1/6$,
    \begin{equation}
      \Cavg =\bigo{1 + tN_C\sqrt{N_L} \left(N_LN_R/N_C \right)^{\delta}}.
    \end{equation}
\end{theorem}

\begin{proof}
As in \cref{sec:qubit_bound}, let $U \coloneqq U(H,t)$, and $\tU$ be the circuit obtained by the $(2k)^{\mathrm{th}}$ order Trotter-Suzuki formula for $U$, where $k$ is chosen below. The entangling capacity of $H$ on a particular input pure state $\ket{\psi}$ in time $t$ is
\begin{align}
    C_{\Delta}(\ket{\psi},t) &= |S_L(U\ketbra{\psi}{\psi} U^\dagger) - S_L(\ketbra{\psi}{\psi})| \nonumber\\
    &= \Big|S_L(U\ketbra{\psi}{\psi} U^\dagger) -S_L(\tU\ketbra{\psi}{\psi}\tU^\dagger) 
    \spliteq +S_L(\tU\ketbra{\psi}{\psi}\tU^\dagger) - S_L(\ketbra{\psi}{\psi})\Big| \nonumber\\
    &\leq |S_L(U\ketbra{\psi}{\psi} U^\dagger) -S_L(\tU\ketbra{\psi}{\psi}\tU^\dagger)| + 2dN_C \nonumber\\
    &\leq 1 + N_L \norm{U\ketbra{\psi}{\psi} U^\dagger - \tU\ketbra{\psi}{\psi} \tU^\dagger}_1 + 2dN_C,
\end{align}
where in the final line we have applied the Fannes-Audenaert inequality \cite{Audenaert_2007}.

Note that
\begin{align}
    &\norm{U \ketbra{\psi}{\psi} U^\dagger - \tU \ketbra{\psi}{\psi} \tU^\dagger}_1 \nonumber\\
   &=\norm{(U-\tU)\ketbra{\psi}{\psi}U^\dagger + \tU \ketbra{\psi}{\psi}(U-\tU)^\dagger}_1 \nonumber\\
    & \le \label{eq:thm5triangle} \norm{(U- \tU) \ketbra{\psi}{\psi}U^\dagger}_1 + \norm{\tU \ketbra{\psi}{\psi} (U - \tU)^\dagger}_1 \\
    & = \label{eq:thm5purenorm1} 2 \sqrt{\brakett{\psi}{(U - \tU)^\dagger(U - \tU)}{\psi}},
\end{align}
where in \cref{eq:thm5triangle}, we have used the triangle inequality, and \cref{eq:thm5purenorm1} follows from the definition of the 1-norm.
For states drawn from $\mu$, applying Jensen's inequality and then using the fact that $\mu$ is a 1-design,
\begin{align}
    & \expectationq*{\ket{\psi} \sim \mu}{2 \sqrt{\brakett*{\psi}{(U - \tU)^\dagger(U - \tU)}{\psi}}} \nonumber\\
    &\leq 2\sqrt{\expectationq*{\ket{\psi} \sim \mu}{\brakett{\psi}{(U - \tU)^\dagger(U - \tU)}{\psi}}} \nonumber\\
    &= \label{eq:thm5design} 2\norm{U- \tU}_{\rm F}.
\end{align}
Thus, we can bound the expected capacity by the Frobenius distance
\begin{equation}
    \Cavg \leq 1 + 2dN_C + 2N_L\norm{U- \tU}_{\rm F}.
\end{equation}

As in \cref{lem:circuit_approx} we first treat the case of time-independent $H$, and then show how the result generalizes to piecewise time-independent $H$.
Using \cref{lem:trotter_frobenius} with
\begin{equation}
M = \ceil*{t \sqrt{N_L}\left(N_LN_R/N_C\right)^{\frac{1}{4k+2}}},
\end{equation}
we obtain
\begin{equation}
    \Cavg = \bigo{1 + tN_C\sqrt{N_L} \left(N_LN_R/N_C \right)^{\frac{1}{4k+2}}}
\end{equation}
since $d \leq  2 \times 5^{k-1}M$ \cite{Trotter_Frob22}.
The result for the time-independent case follows by choosing $k$ large enough that $1/(4k+2) \leq \delta$.

Since $t/M \leq 1/\sqrt{N_L}$ is smaller than the minimum segment width, the Trotter error formula can be applied, and the result holds in the piecewise time-independent case with minimum segment width $1/\sqrt{N_L}$.
We note that this restriction on the segment width arises for technical reasons (as in \cref{eq:segment_restriction}). The number of segments $M$ cannot be chosen to be very large, since this would allow for large $d$, giving a looser bound on the amount of entanglement generated. Our segment width is thus restricted such that $t/M$ is smaller than the segment width without making $M$ very large.

\end{proof}

Though there exist examples of states with high entangling rate with a bottleneck (as in \cref{fig:ghz}),
\cref{thm:entangling_capacity} shows that  in any 1-design $\mu$ there are few such states that can sustain a high entangling rate for long times. In particular, we can show that the entanglement rate for a randomly drawn state over any given period of time is unlikely to be significantly larger than $\bigo{N_C\sqrt{N_L}}$. Specifically, 
applying Markov's inequality and using \cref{thm:entangling_capacity}, we obtain
\begin{equation}
\label{eq:markov_rate}
    \Prob_{\ket\psi \sim \mu} \bracket*{C_{\Delta}(\ket{\psi}, t) \geq \Gamma t} = \bigo*{\frac{1+tN_C\sqrt{N_L} \left(N_LN_R/N_C \right)^{\delta}}{\Gamma t}},
\end{equation}
which is vanishingly small when
\begin{equation}
    \Gamma \gg N_C\sqrt{N_L} \paren*{\frac{N_LN_R}{N_C}}^{\delta}.
\end{equation}
For a state $\ket{\psi}$ drawn from $\mu$, let $\ket{\psi(t)} = U(H,t)\ket{\psi}$ be a state along its trajectory under evolution by $H$.
\Cref{eq:markov_rate} shows that the averaged instantaneous entanglement rate over time,
\begin{equation}
    C_{\Delta}(\ket{\psi(0)}, T) = \frac{1}{T} \int_0^T \frac{\dd S_L(\ketbra{\psi(t)}{\psi(t)})}{\dd t} \dd t,
\end{equation}
is not greater than $\bigo{N_C\sqrt{N_L} \left(N_LN_R/N_C \right)^{\delta}}$ with high probability.

\section{Discussion}
\label{sec:discussion}

In this paper, we showed a lower bound that scales as $\bigomega{N_R^{1-\delta}/\sqrt{N_L}N_C}$ for any $\delta>0$ on the routing time in graphs with a vertex bottleneck.  For the star graph, this provides a lower bound of $\bigomega{\sqrt{N^{1-\delta}}}$ on the routing time. We further showed an optimal routing protocol that saturates this bound in systems of free fermions on the star graph.

For graphs with a vertex bottleneck, the best previous lower bound we are aware of is $\Omega(1)$, following from the Small Incremental Theorem \cite{SIE} or a Lieb-Robinson bound \cite{QSLreview23}. However, unlike those results, a limitation of our work is our introduction of additional assumptions about the allowed Hamiltonian: piecewise time-independence, absence of ancillas, and bounds on the on-site terms $\norm{h_i}$.
A straightforward argument can be used to prove a similar lower bound for routing with time-dependent Hamiltonians for which the Frobenius norm has suitably bounded derivative, by simulating the time-dependent Hamiltonian evolution with a piecewise time-independent evolution, as in \cite{Berry_Childs_Cleve_Kothari_Somma_2014}.
A natural direction for future work is to extend our lower bound to systems with more general Hamiltonians with more general time dependence or unbounded local terms. Our assumption on the piecewise time dependence of the allowed Hamiltonians stems from our use of a Trotterized circuit with Frobenius commutator norm scaling of error in \cref{lem:circuit_approx}. If bounds on the Trotter error for random input states can be extended beyond piecewise time-independent Hamiltonians to time-dependent Hamiltonians whose norm has unbounded derivative, our result could be extended to arbitrary time-dependent Hamiltonian routing as well. Such an approach might also extend our result to Hamiltonians with unbounded local (on-site) terms, since we could move into the interaction picture, removing the on-site terms and adding time-dependence to the non-local terms. We conjecture that a similar lower bound can be proved for more general time-dependent Hamiltonians (whose norms may have large derivatives), and that our assumed bound on the on-site terms can be removed. Another direction for future work is to investigate routing with a piecewise time-independent Hamiltonian of smaller segment width. For technical reasons, our result is restricted to segment width $\propto 1/\sqrt{N_L}.$ It is natural to ask whether routing with a piecewise time-independent Hamiltonian of arbitrarily small segment width is equivalent to routing with arbitrary time-dependent Hamiltonians (i.e., whether $\lim_{\Delta \to 0} \hrtd{G} = \hrt{G}$).

We also assumed that interactions between any two sites in the same partition have bounded norm. This is a natural assumption in systems where interactions between the partitions are physically similar to interactions within the partitions. A possible direction for future work would be to relax this assumption. We expect that this would lead to a looser lower bound on the routing time. Indeed, in systems with ancillas and unbounded interactions within each partition, we can perform routing faster. An example of this is the $W$-state broadcasting-based protocol of \textcite{adv_lim}, which performs Hamiltonian routing in time $\bigo{\sqrt{N}}$ on the vertex barbell graph with $2N+1$ vertices (\cref{fig:vertex_barbell}), obtaining a speedup over the gate-based model (where routing requires depth $\bigtheta{N}$). 
Recent work \cite{yin2024fastaccurateghzencoding} implies the possibility of even faster routing, in $\poly(\log N)$ time, in the same system. However, speedups obtained in this routing model may not always apply since they require the experimentally challenging capability of fast swap operations between qubits and their ancillas. Our techniques provide an $\bigomega{N^{1/2-\delta}}$ lower bound in the more realistic setting of the vertex barbell graph without ancillas.   
Even with bounded 2-site interactions, a related question for future work is to obtain lower bounds that apply when ancillas are permitted.
\begin{figure}[t]
   \centering
\includegraphics[width=\columnwidth]{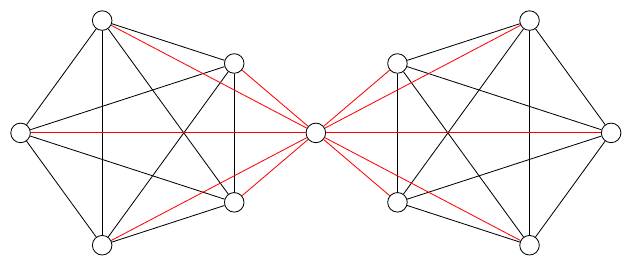}
\caption{The vertex barbell graph with $2N+1$ vertices consists of two copies of the complete graph $K_N$ connected by a common vertex. Here, we depict the vertex barbell with 11 vertices.}
\label{fig:vertex_barbell}
\end{figure}

A more refined analysis might be used to remove the $\delta$ dependence of our qubit lower bound. While our proofs hold for arbitrarily small positive constant $\delta$, we conjecture that the result holds with $\delta=0$. A limitation of our techniques is that our lower bound has a constant factor inversely proportional to $\delta$, so taking the limit as $\delta \rightarrow 0$ causes the lower bound to diverge.

We also showed an upper bound on the average entanglement capacity, or average amount of entanglement that can be produced in time $t$, through a vertex bottleneck. Intuitively, one expects a bottleneck to limit entangling rates, but surprisingly there are states for which fast entanglement with a central qubit is possible. One such example is the $W$ state, which can be used to create constant entanglement with a central vertex in time $\bigo{1/\sqrt{N}}$ \cite{adv_lim}. The GHZ state allows for an even higher entangling rate, producing a constant amount of entanglement in time $\bigo{1/N}$. Our bound on the average entangling capacity for a 1-design shows that states with an entangling rate higher than $\bigomega{\sqrt{N}}$ 
are rare, and constrains their use for tasks such as routing. This result may be interpreted as showing that most states cannot support fast entangling through a vertex bottleneck. An open question is whether this bound can be extended to mixed states as well, which are known to exhibit surprising entanglement dynamics \cite{edss}. It would also be interesting to investigate how this bound constrains other quantum information processing tasks based on their requirements of entanglement. Similarly to our routing lower bound, our entanglement capacity upper bound can be extended to time-dependent Hamiltonians whose norm has bounded derivative by approximating their evolution by a piecewise time-dependent one. Natural questions for future work are to extend the entanglement capacity upper bound to Hamiltonians with arbitrary time dependence, with fast on-site interactions and fast interactions within partitions, and with ancillas. We also conjecture that this bound holds with $\delta = 0$.

The existence of states that saturate the worst-case bounds on entangling rate, which depend on the operator norm, necessitates our proof by the average case, using the Frobenius norm. For quantum information tasks involving operators with significant difference in the operator norm and Frobenius norm (as for nested commutators of Hamiltonians with a vertex bottleneck), our results provide an example where studying the average case yields tighter bounds than the worst case.

Recent work by one of the authors has shown a fast approximate broadcasting (i.e., GHZ encoding) protocol on the complete graph \cite{yin2024fastaccurateghzencoding}. While this protocol has not been proven to have asymptotically vanishing error, numerics suggest this may be the case. Such a protocol, coupled with fast ancilla interactions, would allow for approximate routing on the vertex barbell in time $\poly(\log N)$. An open question is to resolve whether such broadcasting can be done with vanishing error, or even exactly. The best known exact protocol for routing on the vertex barbell (with fast local ancillas) involves broadcasting into a $W$ state \cite{adv_lim}, and takes time $\bigo{\sqrt{N}}$.

While we developed optimal protocols for routing free particles, a natural open question is whether such protocols can be sped up by allowing interaction terms or even terms that do not preserve particle number, such as $l_i^{\dagger} c_i^{\dagger}$. In bosonic systems, introduction of a certain type of interaction term has been shown to provide speedups for the closely related task of state transfer \cite{Vu_Kuwahara_Saito_2024}. 

In \cref{sec:qubit_nogo}, we discussed the difficulty of extending our free-particle routing protocols to work for qubits or qudits. A major open question is to obtain fast (sublinear-time) qubit routing protocols through bottlenecks that leverage continuous time evolution. While such results have been obtained on the barbell graph using $W$ states (and we expect that the result of \textcite{yin2024fastaccurateghzencoding} allows for even faster routing using GHZ states), these protocols require the use of ancillas with fast interactions. Currently, we do not know of any sublinear-time routing protocol for qubits on the star graph.

\section*{Acknowledgments}
We thank Adam Ehrenberg and Joe Iosue for helpful discussions. D.D.\ acknowledges support by the NSF GRFP under Grant No.~DGE-1840340 and an LPS Quantum Graduate Fellowship. C.Y. and A.L. acknowledge support from the Department of Energy under Quantum Pathfinder Grant DE-SC0024324. 
E.S.\ acknowledges support by the Laboratory Directed Research
and Development program of Los Alamos National Laboratory under project number 20210639ECR
and then by the U.S. Department of Energy, Office of Science, National Quantum Information Science Research Centers, Quantum Science Center. D.D., A.Y.G, A.M.C., and A.V.G.~were supported in part by the DoE ASCR Quantum Testbed Pathfinder program (awards No.~DE-SC0019040 and No.~DE-SC0024220), NSF QLCI (award No.~OMA-2120757), and the U.S.~Department of Energy, Office of Science, Accelerated Research in Quantum Computing, Fundamental Algorithmic Research toward Quantum Utility (FAR-Qu). D.D., A.Y.G, and A.V.G.~also acknowledge support from NSF STAQ program, AFOSR MURI, ARL (W911NF-24-2-0107), DARPA SAVaNT ADVENT, NQVL:QSTD:Pilot:FTL. D.D., A.Y.G, and A.V.G.~also acknowledge support from the U.S.~Department of Energy, Office of Science, National Quantum Information Science Research Centers, Quantum Systems Accelerator.


%

\appendix
\onecolumngrid
\section{Proof of Lemma \ref{lem:trotter_frobenius}}
\label{sec:frob_proof}

In this section, we prove \cref{lem:trotter_frobenius}. 
First, we upper bound the Frobenius norms of commutators of terms in the Hamiltonian, which allows us to construct low-depth circuit approximations to the Hamiltonian evolution using the Trotter-Suzuki formula.

\begin{lemma}\label{lem:H<sqrtN_app}
Given a sequence  $\eta_1,\ldots,\eta_{P+Q}\in \{LC,R\}$,
    let the nested commutator $[H_{\eta_{P+Q}}, \ldots, [H_{\eta_2}, H_{\eta_1}]]$ include $H_{LC}$ a total of $P$ times and $H_R$ a total of $Q$ times.
    Then there exists an $f(P,Q) > 0$ with no dependence on $N_L, N_C, N_R$ such that
    \begin{equation}\label{eq:H<sqrtNapp}
        \norm{[H_{\eta_{P+Q}}, \ldots, [H_{\eta_2}, H_{\eta_1}]]}_{\rm F} \le  f(P,Q) \sqrt{N_L^P N_R^QN_C}.
    \end{equation}
\end{lemma}

\begin{proof}
The proof strategy is as follows. We first show that the nested commutator can be expressed in terms of sums over Pauli strings with bounded coefficients, using induction to bound the coefficients. The base case is the first-order commutator. We write out terms of the commutator and regroup coefficients of identical Pauli strings that originate from products of different terms, and then bound the resulting coefficient. We take a similar approach to recursively identify the structure of the higher-order nested commutator and bound its coefficients. Finally, we use orthogonality of Pauli strings under the Frobenius norm to compute the bound of \cref{eq:H<sqrtNapp}.

For convenience, we re-label the terms in $H$:
{\allowdisplaybreaks
\begin{align}
    H^{(L)} &= \sum_{l_1 \in L} h_{l_1}^{(L)\alpha_1} X^{\alpha_1}_{l_1}, &
    H^{(LL)}&= \sum_{\substack{l_1, l_2\in L \\ l_1< l_2}}  h_{l_1, l_2}^{(LL)\alpha_1\alpha_2} X_{l_1}^{\alpha_1} X_{l_2}^{\alpha_2}, &
    H^{(LC)}&= \sum_{\substack{l_1 \in L\\ c_1 \in C}}  h_{l_1, c_1}^{(LC)\alpha_1\gamma_1} X_{l_1}^{\alpha_1} X_{c_1}^{\gamma_1}, \notag \\
    H^{(C)} &= \sum_{c_1 \in C} h_{c_1}^{(C)\gamma_1} X^{\gamma_1}_{c_1}, &
    H^{(CC)}&= \sum_{\substack{c_1, c_2 \in C\\c_1< c_2}}  h_{c_1, c_2}^{(CC)\gamma_1\gamma_2} X_{c_1}^{\gamma_1} X_{c_2}^{\gamma_2}, \notag \\
    H^{(R)} &= \sum_{r_1 \in R} h_{r_1}^{(R)\beta_1} X^{\beta_1}_{r_1}, &
    H^{(RR)}&=\sum_{\substack{r_1, r_2 \in R \\r_1< r_2}}  h_{r_1, r_2}^{(RR)\beta_1\beta_2} X_{r_1}^{\beta_1} X_{r_2}^{\beta_2}, &
    H^{(RC)}&= \sum_{\substack{r_1 \in R \\ c_1 \in C}} h_{r_1, c_1}^{(CR)\beta_1\gamma_1} X_{r_1}^{\beta_1} X_{c_1}^{\gamma_1}
\end{align}
}%
(recall that we use the Einstein summation convention for $\alpha_1, \alpha_2, \beta_1, \beta_2, \gamma_1, \gamma_2 \in\{1,2,3\}$).
Here, $\forall \, i, l_i \in L, r_i \in R, c_i \in C$. We number the vertices in $L,C,R$ uniquely so that these are disjoint sets.
We group these terms as 
\begin{subequations}\begin{align}
    H_{LC} &= H^{(L)} + H^{(LL)} +  H^{(LC)} + H^{(C)} + H^{(CC)},\\
    H_R &= H^{(R)} + H^{(RR)} + H^{(RC)}.
\end{align}\end{subequations}
As mentioned above, we proceed by induction. 
For any sequence $\eta_1,\ldots,\eta_{P+Q}\in\{LC,R\}$, suppose there are $P\ge 1$ $LC$s and $Q\ge 1$ $R$s (the order does not matter). We will show by induction that the corresponding nested commutator, which we call a $(P,Q)$-commutator, is of the form
\begin{align}\label{eq:comm=}
    [H_{\eta_{P+Q}}, \ldots, [H_{\eta_2}, H_{\eta_1}]] &= \sum_{p_1=0}^{P} \sum_{p_2=1}^{P+1-p_1} \sum_{q=0}^{Q} \sum_{\li_1<\li_2<\cdots<\li_{p_1}} \sum_{c_1<c_2<\cdots<c_{p_2}} \sum_{\ri_1<\ri_2<\cdots<\ri_{q}} \notag \\&\qquad \Acoeff X_{\li_1}^{\alpha_1}\cdots X_{\li_{p_1}}^{\alpha_{p_1}} X_{\ri_1}^{\beta_1}\cdots X_{\ri_q}^{\beta_q} X_{c_1}^{\gamma_1} \cdots X_{c_{p_2}}^{\gamma_{p_2}}, \end{align}
    where
    \begin{equation}
     \abs{\Acoeff} \le c_{P,Q} {N_L}^{(P-p_1)/2}{N_R}^{(Q-q)/2}{N_C}^{(1-p_2)/2}, \label{eq:induction_A}
\end{equation}
and where $c_{P,Q}$ is a constant that does not depend on $N_L, N_C, N_R$. 

The base case is the first-order commutator 
\begin{align}
	\label{eq:comm_first_order}
    [H_{LC}, H_R] &= \sum_{c_1, r_1, c_1'} h_{c_1}^{(C)\gamma_1} h_{r_1, c_1'}^{(CR)\beta_1\gamma_1'} X_{r_1}^{\beta_1}[ X^{\gamma_1}_{c_1}, X_{c_1'}^{\gamma_1'}  ]
    + \sum_{c_1, c_2, r_1, c_1'} h_{c_1, c_2}^{(CC)\gamma_1\gamma_2} h_{r_1, c_1'}^{(CR)\beta_1\gamma_1'} X_{r_1}^{\beta_1} [X_{c_1}^{\gamma_1} X_{c_2}^{\gamma_2}, X_{c_1'}^{\gamma_1'} ] \nonumber \\
    &\qquad+ \sum_{l_1, c_1, r_1, c_1'} h_{l_1, c_1}^{(LC)\alpha_1\gamma_1} h_{r_1, c_1'}^{(CR)\beta_1\gamma_1'} X_{r_1}^{\beta_1}  X_{l_1}^{\alpha_1} [X_{c_1}^{\gamma_1}, X_{c_1'}^{\gamma_1'}] .
\end{align}
By expanding the commutators and grouping terms based on which sites they are supported (i.e., act non-trivially) on, we can rewrite this in the form
\begin{equation}
\label{eq:grouped_first_order}
 [H_{LC}, H_R] =\sum_{r_1, c_1} A^{\beta_1 \tilde{\gamma_1}}_{r_1 c_1} X_{r_1}^{\beta_1} X_{c_1}^{\tilde{\gamma_1}}
+ \sum_{r_1, c_1, c_2} A^{\beta_1 \tilde{\gamma_1} \tilde{\gamma_2}}_{r_1 c_1 c_2} X_{r_1}^{\beta_1} X_{c_1}^{\tilde{\gamma_1}} X_{c_2}^{\tilde{\gamma_2}}
+ \sum_{l_1, r_1, c_1} A^{\alpha_1 \beta_1 \tilde{\gamma_1}}_{l_1 r_1 c_1} X_{l_1}^{\alpha_1} X_{r_1}^{\beta_1}  X_{c_1}^{\tilde{\gamma_1}}.
\end{equation}
$A^{\beta_1 \tilde{\gamma_1}}_{r_1 c_1}, A^{\alpha_1 \beta_1 \tilde{\gamma_1}}_{l_1 r_1 c_1}, A^{\beta_1 \tilde{\gamma_1} \tilde{\gamma_2}}_{r_1 c_1 c_2}$ are coefficients of distinct Pauli strings in \cref{eq:grouped_first_order}. For conciseness, we leave $P,Q$ implicit in our notation for the first-order commutator, where $P=Q=1$. Also note that, for example, the coefficients $A^{\alpha_1 \beta_1 \tilde{\gamma_1}}_{l_1 r_1 c_1}$ and $ A^{\beta_1 \tilde{\gamma_1} \tilde{\gamma_2}}_{r_1 c_1 c_2}$ are distinguishable since the lower indices come from distinct sets (e.g., $l_1 \in L$ and $r_1 \in R$).
Each $A^{\beta_1 \tilde{\gamma_1}}_{r_1 c_1}$ is a sum of coefficients of $X_{r_1}^{\beta} X_{c_1}^{\tilde{\gamma_1}}$ from \cref{eq:comm_first_order}: \begin{equation}\label{eq:A_example}
    A^{\beta_1 \tilde{\gamma_1}}_{r_1 c_1} = 2\sum_{\gamma_1,\gamma_1'} h_{c_1}^{(C)\gamma_1} h_{r_1, c_1}^{(CR)\beta_1\gamma_1'} \mathbb{I}[X^{\gamma_1}_{c_1}X^{\gamma'_1}_{c_1}=X^{\tilde{\gamma_1}}_{c_1}],
\end{equation}
where $\mathbb{I}[x]$ is the indicator function that returns $0$ ($1$) if $x$ is false (true). For each $A^{\beta_1 \tilde{\gamma_1}}_{r_1 c_1}$, there are four contributing terms from \cref{eq:comm_first_order} (counting the multiplicity factor $2$ in \cref{eq:A_example}), obtained by expanding the commutator and making use of the Pauli multiplication rules. The number of such terms is the number of valid indices $(c_1, r_1, c_1', \beta_1, \gamma_1, \gamma_1')$ that, on regrouping \cref{eq:comm_first_order} to the form \cref{eq:grouped_first_order}, contribute to terms with new indices $(c_1, r_1, \beta_1, \tilde{\gamma_1})$. A tilde on an index indicates that it is an updated index in the commutator formed by the multiplication of two Paulis as in \cref{eq:A_example}, where $\tilde{\gamma_1}$ is the index corresponding to the single-qubit Pauli obtained on multiplying $X^{\gamma_1}$ and $X^{\gamma_1'}$. Indices without a tilde are unchanged and are identical to those in $H_{LC}, H_R.$ 

Likewise, each $A^{\alpha_1 \beta_1 \tilde{\gamma_1}}_{l_1 r_1 c_1}$ is the sum of four coefficients from \cref{eq:comm_first_order}, and each $A^{\beta_1 \tilde{\gamma_1} \tilde{\gamma_2}}_{r_1 c_1 c_2}$ is the sum of eight coefficients.

Hence \begin{subequations}
\begin{align}
    |A^{\beta_1 \tilde{\gamma_1}}_{r_1 c_1}| & \leq 4 \max |h_{c_1}^{(C)\gamma_1} h_{r_1, c_1}^{(CR)\beta_1\gamma_1}| \leq 4\sqrt{N}, \label{eq:factor4} \\
    |A^{\alpha_1 \beta_1 \tilde{\gamma_1}}_{l_1 r_1 c_1}| & \leq 4 \max |h_{l_1, c_1}^{(LC)\alpha_1\gamma_1} h_{r_1, c_1}^{(CR)\beta_1\gamma_1}| \leq 4, \\
    |A^{\beta_1 \tilde{\gamma_1} \tilde{\gamma_2}}_{r_1 c_1 c_2}| & \leq 8 \max |h_{c_1, c_2}^{(CC)\gamma_1\gamma_2} h_{r_1, c_1'}^{(CR)\beta_1\gamma_1'}| \leq 8.
\end{align}\end{subequations}
Now, as different Pauli strings are orthogonal with respect to the Frobenius norm, we have 
\begin{align}
    \norm{[H_{LC}, H_R]}_{\rm F}^2 &\leq \sum_{c_1, r_1} 16N
    + \sum_{c_1, c_2, r_1} 64
    + \sum_{l_1, c_1, r_1} 16\notag \\
    &\leq 16N \times N_C \times N_R + 64 N_C^2 \times N_R + 16 N_L \times N_C \times N_R,
\end{align}
satisfying \cref{eq:H<sqrtNapp} for $P=Q=1$ because we have assumed $N_L$ is larger than $N_R,N_C$.

We now consider higher-order commutators. 
We use the induction hypothesis \cref{eq:comm=} to prove that the same equation holds for the $(P+1,Q)$-commutator, i.e., adding one more outermost commutator with $H_{LC}$. The case $\widetilde{P}=P$, $\widetilde{Q}=Q+1$ follows analogously.  We use a tilde to denote next-order quantities, e.g., the numbers of $LC$s and $R$s in the new $\eta$ sequence are \begin{align}\label{eq:tPtQ}
    \widetilde{P}&=P+1,  & \widetilde{Q}&=Q.
\end{align}

We use $\cdots$ to indicate products over intermediate indices. For example, $X_{\li_1}^{\alpha_1}\cdots X_{\li_{j}}^{\alpha_j}$ indicates $\prod_{i=1}^{j} X_{\li_i}^{\alpha_i}$. We also use $X^0_j$  to denote the identity operator for any $j$. We write $l\notin \overrightarrow{\li_{p_1}}$ as shorthand for $l\notin\{\li_1,\ldots,\li_{p_1}\}$, and similarly for $c\notin \overrightarrow{c_{p_2}}$, etc.
We also abbreviate $\Acoeff$ as $\Acoeffvec$.
We first compute the commutators of terms in $H_{LC}$ (whose indices are labeled by primes, e.g., $X_{l_1'}^{\alpha_1'} X_{c_1'}^{\gamma_1'}$) with terms from the $(P,Q)$-commutator (whose indices are labeled without primes, e.g., $X_{\li_1}^{\alpha_1}\cdots X_{\li_{p_1}}^{\alpha_{p_1}} X_{\ri_1}^{\beta_1}\cdots X_{\ri_q}^{\beta_q} X_{c_1}^{\gamma_1} \cdots X_{c_{p_2}}^{\gamma_{p_2}}$).

First, we consider the commutator with $H^{(L)}$.
Any term $X_{l_1'}^{\alpha_1'}$ in $H^{(L)}$ is non-commuting with $X_{\li_1}^{\alpha_1}\cdots X_{\li_{p_1}}^{\alpha_{p_1}} X_{\ri_1}^{\beta_1}\cdots X_{\ri_q}^{\beta_q} X_{c_1}^{\gamma_1} \cdots X_{c_{p_2}}^{\gamma_{p_2}}$ iff $\exists \, j$ such that $l_j = l_1'$ and $\alpha_j\neq \alpha_1'$. Therefore,
\begin{align}
    &[H^{(L)}, \Acoeffvec  X_{\li_1}^{\alpha_1}\cdots X_{\li_{p_1}}^{\alpha_{p_1}} X_{\ri_1}^{\beta_1}\cdots X_{\ri_q}^{\beta_q} X_{c_1}^{\gamma_1} \cdots X_{c_{p_2}}^{\gamma_{p_2}}] \nonumber \\
    &\label{eq:pqa} \qquad= \sum_{l_1' \in L} h_{l_1'}^{(L)\alpha_1'} \Acoeffvec \mathbb{I}[\li_1' =\li_j] X_{\li_1}^{\alpha_1}\cdots X_{\li_{j}}^{0} \cdots X_{\li_{p_1}}^{\alpha_{p_1}} X_{\ri_1}^{\beta_1}\cdots X_{\ri_q}^{\beta_q} X_{c_1}^{\gamma_1} \cdots X_{c_{p_2}}^{\gamma_{p_2}}[X^{\alpha_1'}_{l_1'}, X^{\alpha_j}_{l_j}]. 
\end{align}
Now we consider terms in $H^{(LL)}$. For a term $X_{l_1'}^{\alpha_1'}X_{l_2'}^{\alpha_2'}$, we divide the terms from the $(P,Q)$-commutator which have the form $X_{\li_1}^{\alpha_1}\cdots X_{\li_{p_1}}^{\alpha_{p_1}}X_{\ri_1}^{\beta_1}\cdots X_{\ri_q}^{\beta_q} X_{c_1}^{\gamma_1} \cdots X_{c_{p_2}}^{\gamma_{p_2}}$ into those with exactly one overlapping site with $l_1',l_2'$ or exactly two overlapping sites. The commutator is thus expressed as
\begin{align}
    &[H^{(LL)}, \Acoeffvec X_{\li_1}^{\alpha_1}\cdots X_{\li_{p_1}}^{\alpha_{p_1}} X_{\ri_1}^{\beta_1}\cdots X_{\ri_q}^{\beta_q} X_{c_1}^{\gamma_1} \cdots X_{c_{p_2}}^{\gamma_{p_2}}] \nonumber \\ 
    &\qquad=  \sum_{\substack{l_1, l_2\in L \\ l_1< l_2}}  h_{l_1', l_2'} ^{(LL)\alpha_1'\alpha_2'} \Acoeffvec X_{\ri_1}^{\beta_1}\cdots X_{\ri_q}^{\beta_q} X_{c_1}^{\gamma_1} \cdots X_{c_{p_2}}^{\gamma_{p_2}} \times \nonumber\\
    &\qquad\qquad\Bigg(\label{eq:pqb}
    \mathbb{I}[\li_1' =\li_j]\mathbb{I}[\li_2' \notin \overrightarrow{\li_{p_1}}] X_{\li_1}^{\alpha_1}\cdots   X^{0}_{l_{j}} \cdots X_{\li_{p_1}}^{\alpha_{p_1}}  X_{l_2'}^{\alpha_2'}[X^{\alpha_1'}_{l_1'}, X^{\alpha_j}_{l_j}] \\
    &\label{eq:pqc}\qquad\qquad+\mathbb{I}[\li_1' \notin \overrightarrow{\li_{p_1}}]\mathbb{I}[\li_2' = \li_j] X_{\li_1}^{\alpha_1}\cdots X^{0}_{l_{j}}  \cdots X_{\li_{p_1}}^{\alpha_{p_1}}  X_{l_1'}^{\alpha_1'} [X_{l_2'}^{\alpha_2'}, X^{\alpha_j}_{l_j}] \\
    &\label{eq:pqd}\qquad\qquad+ \mathbb{I}[\li_1' = \li_{j_1}]\mathbb{I}[\li_2' = \li_{j_2}] X_{\li_1}^{\alpha_1}\cdots X^{0}_{l_{j_1}}  \cdots X^{0}_{l_{j_2}} \cdots X_{\li_{p_1}}^{\alpha_{p_1}}[X_{l_1'}^{\alpha_1'}X_{l_2'}^{\alpha_2'}, X_{\li_{j_1}}^{\alpha_{j_1}} X_{\li_{j_2}}^{\alpha_{j_2}}]  \Bigg).
\end{align}
Now we look at the commutator with $H^{(LC)}$. Similarly to $H^{(LL)}$, we divide terms of the order-$(P+Q)$ nested commutator into those that have exactly one or two overlapping sites with each term of $H^{(LC)}$. Therefore,
\begin{align}
    &[H^{(LC)}, \Acoeffvec X_{\li_1}^{\alpha_1}\cdots X_{\li_{p_1}}^{\alpha_{p_1}} X_{\ri_1}^{\beta_1}\cdots X_{\ri_q}^{\beta_q} X_{c_1}^{\gamma_1} \cdots X_{c_{p_2}}^{\gamma_{p_2}}]  = \sum_{\substack{l_1 \in L\\ c_1 \in C}} h_{l_1', c_1'}^{(LC)\alpha_1'\gamma_1'} \Acoeffvec X_{\ri_1}^{\beta_1}\cdots X_{\ri_q}^{\beta_q} \times \nonumber \\
    &\qquad\Bigg(
    \mathbb{I}[\li_1' = \li_j]\mathbb{I}[c_1' = c_k] X_{\li_1}^{\alpha_1}\cdots X_{\li_j}^{0} \cdots X_{\li_{p_1}}^{\alpha_{p_1}}  X_{c_1}^{\gamma_1} \cdots X_{c_k}^{0} \cdots X_{c_{p_2}}^{\gamma_{p_2}} [X_{l_1'}^{\alpha_1'}X_{c_1'}^{\gamma_1'}, X^{\alpha_j}_{l_j}X_{c_k}^{\gamma_k}] \label{eq:pqe} \\
    &\qquad\qquad+\mathbb{I}[\li_1' \notin \overrightarrow{\li_{p_1}}]\mathbb{I}[c_1' = c_j] X_{\li_1}^{\alpha_1}\cdots X_{\li_{p_1}}^{\alpha_{p_1}} X_{l_1'}^{\alpha_1'}  X_{c_1}^{\gamma_1} \cdots X^{0}_{c_{j}} \cdots X_{c_{p_2}}^{\gamma_{p_2}}  [X_{c_1'}^{\gamma_1'},X_{c_j}^{\gamma_j}]\label{eq:pqf} \\
    &\qquad\qquad\qquad+\mathbb{I}[\li_1' = \li_j]\mathbb{I}[c_1' \notin \overrightarrow{c_{p_2}}] X_{\li_1}^{\alpha_1}\cdots  X^{0}_{l_j}   \cdots X_{\li_{p_1}}^{\alpha_{p_1}} X_{c_1}^{\gamma_1} \cdots X_{c_{p_2}}^{\gamma_{p_2}}  X_{c_1'}^{\gamma_1'}[X_{l_1'}^{\alpha_1'}, X^{\alpha_j}_{l_j}] \label{eq:pqg}
    \Bigg).
\end{align}
Now we consider the commutator with $H^{(C)}$, which is analogous to the commutator with $H^{(L)}$:
\begin{align}
&[H^{(C)}, \Acoeffvec X_{\li_1}^{\alpha_1}\cdots X_{\li_{p_1}}^{\alpha_{p_1}} X_{\ri_1}^{\beta_1}\cdots X_{\ri_q}^{\beta_q} X_{c_1}^{\gamma_1} \cdots X_{c_{p_2}}^{\gamma_{p_2}}] \nonumber \\
&\label{eq:pqh}\qquad=\sum_{c_1'\in C} h_{c_1'}^{(C)\gamma_1'} \Acoeffvec \mathbb{I}[c_1' =c_j] X_{\li_1}^{\alpha_1} \cdots X_{\li_{p_1}}^{\alpha_{p_1}} X_{\ri_1}^{\beta_1}\cdots X_{\ri_q}^{\beta_q} X_{c_1}^{\gamma_1} \cdots  X^{0}_{c_{j}} \cdots X_{c_{p_2}}^{\gamma_{p_2}} [X^{\gamma_1'}_{c_1'},X_{c_j}^{\gamma_j}].
\end{align}

The case for $H^{(CC)}$ is analogous to the commutator with $H^{(LL)}$:
\begin{align}
&[H^{(CC)}, \Acoeffvec X_{\li_1}^{\alpha_1}\cdots X_{\li_{p_1}}^{\alpha_{p_1}} X_{\ri_1}^{\beta_1}\cdots X_{\ri_q}^{\beta_q} X_{c_1}^{\gamma_1} \cdots X_{c_{p_2}}^{\gamma_{p_2}}] \nonumber \\
&\qquad= \sum_{\substack{c_1, c_2 \in C\\c_1< c_2}}  h_{c_1', c_2'}^{(CC)\gamma_1'\gamma_2'} \Acoeffvec X_{\li_1}^{\alpha_1} \cdots X_{\li_{p_1}}^{\alpha_{p_1}} X_{\ri_1}^{\beta_1}\cdots X_{\ri_q}^{\beta_q} \times \nonumber \\
&\label{eq:pqi}\qquad\qquad\Bigg( \mathbb{I}[c_1' =c_j]\mathbb{I}[c_2' \notin \overrightarrow{c_{p_2}}] X_{c_1}^{\gamma_1} \cdots X^{0}_{c_{j}}   \cdots X_{c_{p_2}}^{\gamma_{p_2}}  X_{c_2'}^{\gamma_2'} [X_{c_1'}^{\gamma_1'}, X_{c_j}^{\gamma_j}] \\
&\label{eq:pqj}\qquad\qquad+ \mathbb{I}[c_1' \notin \overrightarrow{c_{p_2}}]\mathbb{I}[c_2' = c_j] X_{c_1}^{\gamma_1} \cdots X^{0}_{c_{j}}  \cdots X_{c_{p_2}}^{\gamma_{p_2}}   X_{c_1'}^{\gamma_1'}[X_{c_2'}^{\gamma_2'},X_{c_j}^{\gamma_j}]  \\
&\label{eq:pqk}\qquad\qquad+ \mathbb{I}[c_1' = c_{j_1}]\mathbb{I}[c_2' = c_{j_2}] X_{c_1}^{\gamma_1} \cdots X_{c_{j_1}}^0 \cdots X_{c_{j_2}}^0  \cdots X_{c_{p_2}}^{\gamma_{p_2}} [ X_{c_1'}^{\gamma_1'} X_{c_2'}^{\gamma_2'},X_{c_{j_1}}^{\gamma_{j_1}} X_{c_{j_2}}^{\gamma_{j_2}}]  \Bigg).
\end{align}

Overall, the $(P+1,Q)$-commutator is
{\allowdisplaybreaks
\begin{align}
\label{eq:p1qcommutator}
    &[H_{LC}, [H_{\eta_{P+Q}}, \cdots, [H_{\eta_2}, H_{\eta_1}]]] = \sum_{p_1=0}^{P} \sum_{p_2=1}^{P+1-p_1} \sum_{q=0}^{Q} \sum_{\li_1<\li_2<\cdots<\li_{p_1}} \sum_{c_1<c_2<\cdots<c_{p_2}} \sum_{\ri_1<\ri_2<\cdots<\ri_{q}} \Bigg{(} \Acoeffvec \nonumber \\
    &\qquad \times [H^{(L)}+H^{(LL)}+H^{(LC)}+H^{(C)}+H^{(CC)}, X_{\li_1}^{\alpha_1}\cdots X_{\li_{p_1}}^{\alpha_{p_1}} X_{\ri_1}^{\beta_1}\cdots X_{\ri_q}^{\beta_q} X_{c_1}^{\gamma_1} \cdots X_{c_{p_2}}^{\gamma_{p_2}}] \Bigg{)}.
\end{align}
}
Using \crefrange{eq:pqa}{eq:pqk} and grouping terms together, we can rewrite this as
\begin{multline}
\label{eq:nested_commutator_grouped}
    [H_{LC}, [H_{\eta_{P+Q}}, \cdots, [H_{\eta_2}, H_{\eta_1}]]]= \sum_{\widetilde{p_1}=0}^{P+1} \sum_{\widetilde{p_2}=1}^{P+2-\widetilde{p_1}} \sum_{q=0}^{Q} \\ \qquad\qquad \sum_{\widetilde{\li_1}<\cdots<\widetilde{\li}_{\widetilde{p_1}}} \sum_{\widetilde{c}_1<\cdots<\widetilde{c}_{\widetilde{p_2}}} \sum_{\ri_1<\ri_2<\cdots<\ri_{q}} \Atilde \\
    X_{\widetilde{\li_1}}^{\widetilde{\alpha_1}}\cdots X_{\li_{\widetilde{p_1}}}^{\widetilde{\alpha_{\widetilde{p_1}}}} X_{\ri_1}^{\widetilde{\beta_1}}\cdots X_{\ri_q}^{\widetilde{\beta_q}} X_{\widetilde{c_1}}^{\widetilde{\gamma_1}} \cdots X_{c_{\widetilde{p_2}}}^{\widetilde{\gamma_{\widetilde{p_2}}}}.
\end{multline}

Each coefficient $\Atilde$ is a sum of coefficients of the form $\Acoeffvec \times h$, where $h$ is a coefficient from $H_{LC}$. We can thus bound $\norm{\Atilde}$ using the induction hypothesis, which bounds $\norm{\Acoeff}$, and by counting the number of terms (and their norms) in each of \crefrange{eq:pqa}{eq:pqk} that contribute to the $(P+1, Q)$ commutator:
\begin{itemize}
    \item In \cref{eq:pqa}, the number of sites in the support of the commutator is the same as the number of sites in the support of $X_{\li_1}^{\alpha_1}\cdots X_{\li_{p_1}}^{\alpha_{p_1}} X_{\ri_1}^{\beta_1}\cdots X_{\ri_q}^{\beta_q} X_{c_1}^{\gamma_1} \cdots X_{c_{p_2}}^{\gamma_{p_2}}$, i.e., $\widetilde{p_1}=p_1,\widetilde{p_2}=p_2$. We count the number of terms in the $(P+1,Q)$ commutator that are obtained from \cref{eq:pqa}. For a given $\overrightarrow{\widetilde{l_{\widetilde{p_1}}}}$, we count the number of terms from \cref{eq:pqa} with  $l_1',\overrightarrow{l_{p_1}}$ that would produce a term in the $(P+1,Q)$-commutator with $\overrightarrow{\widetilde{l_{\widetilde{p_1}}}}$. We see that $\overrightarrow{l_{p_1}}$ must equal $\overrightarrow{\widetilde{l_{\widetilde{p_1}}}}$, and there are $\widetilde{p_1}$ choices of $\li'_1$ that overlap with a site from $\overrightarrow{l_{p_1}}$.  As in \cref{eq:factor4}, we have a multiplicity of 4, obtained by expanding the commutator and from the Pauli multiplication rules. Thus, there are at most $4\widetilde{p_1}$ terms from \cref{eq:pqa} which are grouped together in \cref{eq:nested_commutator_grouped} for each coefficient $\Atilde$.
    
    \item Similar to the analysis of \cref{eq:pqa}, we count the number of terms of \cref{eq:pqb} that contribute to a given $\overrightarrow{\widetilde{l_{\widetilde{p_1}}}}$. In \cref{eq:pqb}, the number of sites in $L$ in the support of the commutator increases by 1, so we have $\widetilde{p_1} = p_1+1, \widetilde{p_2} = p_2$. There are ${\widetilde{p_1}}$ choices of $\li_2'$ and ${\widetilde{p_1}}-1$ choices of $\li_1'$ that match a given $\overrightarrow{\widetilde{l_{\widetilde{p_1}}}}$. Hence there are at most $4\widetilde{p_1}(\widetilde{p_1}-1)$ coefficients from \cref{eq:p1qcommutator} that are grouped together in $\Atilde$.
    
    \item \Cref{eq:pqc} is similar to \cref{eq:pqb}. There are at most $4\widetilde{p_1}(\widetilde{p_1}-1)$ coefficients from \cref{eq:p1qcommutator} that are grouped together in $\Atilde$.
    
    \item \Cref{eq:pqd} is non-zero on terms of $H^{(LL)}$ which have two overlapping sites with terms of the $(P,Q)$-commutator, such that the Paulis commute on one site and anticommute on the other site. An example of this on two sites is $[X_1X_2,X_1Z_2]\propto Y_2$. Once again, we count the number of choices of sites $l_1', l_2', \overrightarrow{l_{p_1}}$ in \cref{eq:pqd} that contribute to a given $\Atilde$, i.e., match the sites $\overrightarrow{\widetilde{l_{\widetilde{p_1}}}}$. There are at most $N_L$ choices of the site to be removed ($\li_2'$), and $\widetilde{p_1}$ choices of the site to retain $(\li_1'$).
    For a given Pauli on site $\li_1'$, there are 6 commutators that contribute to the coefficient. For example, $Y_2$ can be obtained from the following commutators: $[X_1X_2,X_1Z_2],[Y_1X_2,Y_1Z_2],[Z_1X_2,Z_1Z_2]$ and the commutators obtained from these by swapping $X_2 \leftrightarrow Z_2$. On expanding each commutator, we obtain two terms. There are thus at most $12N_L\widetilde{p_1}$ contributing terms from \cref{eq:pqd} grouped in $\Atilde$ with $\widetilde{p_1} = p_1-1, \widetilde{p_2} = p_2$. 
    
    \item \Cref{eq:pqe} is similar to \cref{eq:pqd}. Terms in \cref{eq:pqe} have one overlapping site in $C$ and one overlapping site in $L$ with the $(P,Q)$-commutator. 
    Similar to \cref{eq:pqd}, the number of sites decreases, in either $C$ or $L$. If the site removed is from $C$, then there are at most $N_C$ choices of the site to be removed ($c_1'$), and $\widetilde{p_1}$ choices of the site to retain $(\li_1'$). Likewise, if the site removed is from $L$, then there are $N_L$ choices of $\li_1'$ and $\widetilde{p_2}$ choices of $c_1'$. 
    There are at most $12N_L\widetilde{p_2}$ terms of \cref{eq:pqe} grouped in $\Atilde$ with $\widetilde{p_1} = p_1-1, \widetilde{p_2} = p_2$.
    Likewise, there are at most $12N_C\widetilde{p_1}$ terms of \cref{eq:pqe} grouped in $\Atilde$ with $\widetilde{p_1} = p_1, \widetilde{p_2} = p_2-1$.
    
    \item From \cref{eq:pqf}, the number of sites in $L$ in the support of the commutator increases by one, while the number of sites in the support from $C$ remains unchanged. Therefore, there are at most $4\widetilde{p_1}\widetilde{p_2}$ terms of \cref{eq:pqf} grouped in $\Atilde$ with $\widetilde{p_1} = p_1+1, \widetilde{p_2} = p_2$. 
    
    \item \Cref{eq:pqg} is similar to \cref{eq:pqf}, increasing the number of sites from $C$ while leaving the number of sites from $L$ unchanged. There are at most $4\widetilde{p_1}\widetilde{p_2}$ terms of \cref{eq:pqf} grouped in $\Atilde$ with $\widetilde{p_1} = p_1, \widetilde{p_2} = p_2+1$. 
    
    \item \Cref{eq:pqh} is similar to \cref{eq:pqa}. There are at most $4\widetilde{p_2}$ terms of \cref{eq:pqh} grouped in $\Atilde$ with $\widetilde{p_1} = p_1, \widetilde{p_2} = p_2$.
    
    \item The analysis of \cref{eq:pqi} is similar to that of \cref{eq:pqb}. There are at most $4\widetilde{p_2}(\widetilde{p_2}-1)$ terms of \cref{eq:pqb} grouped in $\Atilde$ with  $\widetilde{p_1} = p_1, \widetilde{p_2} = p_2+1$.
    
    \item The analysis of \cref{eq:pqj} is similar to that of \cref{eq:pqi}. There are at most $4\widetilde{p_2}(\widetilde{p_2}-1)$ terms of \cref{eq:pqb} grouped in $\Atilde$ with  $\widetilde{p_1} = p_1, \widetilde{p_2} = p_2+1$.
    
    \item The analysis of \cref{eq:pqk} is similar to that of \cref{eq:pqd}. \Cref{eq:pqk} is non-zero on terms of $H^{(CC)}$ which have two overlapping sites with terms of the $(P,Q)$-commutator, such that the Paulis on one site anticommute, and on the other site they commute.  We count the number of choices of sites $c_1', c_2', \overrightarrow{c_{p_2}}$ that match with a given $\overrightarrow{\widetilde{c_{\widetilde{p_2}}}}$. There are at most $N_C$ choices of the site to remove $(c_2')$ and $\widetilde{p_2}$ choices of the site to retain $(c_1')$.  Hence there are at most $12N_C\widetilde{p_2}$ terms of \cref{eq:pqk} grouped in $\Atilde$ with $\widetilde{p_2} = p_2-1, \widetilde{p_1} = p_1$. 
\end{itemize}

\Cref{eq:nested_commutator_grouped} is the sum of \crefrange{eq:pqa}{eq:pqk}. Applying the triangle inequality to the cases above, we thus have
\begin{align}
    \norm{\Atilde} &\leq 4\widetilde{p_1} \max\, \norm{\Acoeffvec} \bigg |_{p_1=\widetilde{p_1},p_2=\widetilde{p_2}} + \cdots \nonumber\\
    &\leq 4\widetilde{p_1}\times c_{P,Q} {N_L}^{(P-\widetilde{p_1})/2}{N_R}^{(Q-q)/2}{N_C}^{(1-\widetilde{p_2})/2} + \cdots,
\end{align}
where the first term comes from \cref{eq:pqa}, and $\cdots$ represents the contributions from the remaining equations. The second line follows from the induction hypothesis [\cref{eq:induction_A}]. Adding the terms from the analysis of each of \crefrange{eq:pqa}{eq:pqk}, we thus have 
\begin{align}
    \norm{\Atilde} &\leq \bigg[\left( 4\widetilde{p_1}+4\widetilde{p_2}\right) +\left(8\widetilde{p_1}^2 + 4\widetilde{p_1}\widetilde{p_2}-8\widetilde{p_1}\right)\sqrt{N_L} \notag \\ &\qquad +\left(8\widetilde{p_2}^2 + 4\widetilde{p_1}\widetilde{p_2}- 8\widetilde{p_2}\right)\sqrt{N_C} + \left(12\widetilde{p_2} + 12\widetilde{p_1}\right)\frac{N_C}{\sqrt{N_C}}  + \left(12\widetilde{p_1} + 12\widetilde{p_2}\right)\frac{N_L}{\sqrt{N_L}}\bigg]  \nonumber \\ & \qquad \times c_{P,Q} {N_L}^{(P-\widetilde{p_1})/2}{N_R}^{(Q-q)/2}{N_C}^{(1-\widetilde{p_2})/2} \nonumber \\
    &\leq \bigg[\left( 4\widetilde{p_1}+4\widetilde{p_2}\right) +\left(8(P+1)^2 + 4P(P+1)-8\widetilde{p_1}\right)\sqrt{N_L} \notag \\ &\qquad +\left(8P^2 + 4P(P+1)- 8\widetilde{p_2}\right)\sqrt{N_C} + \left(12(P+1) + 12P\right)\frac{N_C}{\sqrt{N_C}} \nonumber \\ & \qquad+ \left(12(P+1) + 12P\right)\frac{N_L}{\sqrt{N_L}}\bigg]  \nonumber \\ & \qquad \times c_{P,Q} {N_L}^{(P-\widetilde{p_1})/2}{N_R}^{(Q-q)/2}{N_C}^{(1-\widetilde{p_2})/2} \nonumber \\
    &\leq (24P^2 + 72P +32)\times c_{P,Q} {N_L}^{(P+1-\widetilde{p_1})/2}{N_R}^{(Q-q)/2}{N_C}^{(1-\widetilde{p_2})/2}
    \nonumber\\
    &\leq  c_{P+1,Q} {N_L}^{(P+1-\widetilde{p_1})/2}{N_R}^{(Q-q)/2}{N_C}^{(1-\widetilde{p_2})/2} ,
\end{align}
where, for example, the term $\propto 4\widetilde{p_1}\widetilde{p_2}\sqrt{N_L}$ comes from the contribution of \cref{eq:pqf}, with the extra factor $\sqrt{N_L}$ due to $\widetilde{p_1}=p_1+1$.
The induction hypothesis, \cref{eq:induction_A}, is thus satisfied for $P+1,Q$ by choosing $c_{P+1,Q} \geq (24P^2 + 72P +32) c_{P,Q}$.

Now, we use this to show a bound on the commutator norm. Since different Pauli strings are orthogonal, the Frobenius norm is bounded as
\begin{align}
    \norm{[H_{\eta_{P+Q}}, \cdots, [H_{\eta_2}, H_{\eta_1}]]}_{\rm F}^2  &\leq  \sum_{p_1=0}^{P} \sum_{p_2=1}^{P+1-p_1} \sum_{q=0}^{Q} \sum_{\li_1<\cdots<\li_{p_1}} \sum_{c_1<\cdots<c_{p_2}} \sum_{\ri_1<\ri_2<\cdots<\ri_{q}} \norm{A_{\li_1\cdots \li_{p_1}, c_1 \cdots c_{p_2},\ri_1\cdots \ri_{q}; P+1,Q}^{\alpha_1 \cdots \alpha_{p_1}, \beta_1\cdots \beta_{q},\gamma_1 \cdots \gamma_{p_2}}}^2 \\
    & \leq  \sum_{p_1=0}^{P} \sum_{p_2=1}^{P+1-p_1} \sum_{q=0}^{Q} \sum_{\li_1<\cdots<\li_{p_1}} \sum_{c_1<\cdots<c_{p_2}} \sum_{\ri_1<\ri_2<\cdots<\ri_{q}} c_{P+1,Q}^2 {N_L}^{P-p_1} {N_R}^{Q-q} {N_C}^{1-p_2} \\
    &\leq  \sum_{p_1=0}^{P} \sum_{p_2=1}^{P+1-p_1} \sum_{q=0}^{Q} \binom{N_L}{p_1} \binom{N_C}{p_2} \binom{N_R}{q} c_{P+1,Q}^2 {N_L}^{P-p_1} {N_R}^{Q-q} {N_C}^{1-p_2} \\
    &\leq  \sum_{p_1=0}^{P} \sum_{p_2=1}^{P+1-p_1} \sum_{q=0}^{Q} \frac{N_L^{p_1}}{p_1!} \frac{N_C^{p_2}}{p_2!} \frac{N_R^q}{q!} c_{P+1,Q}^2 {N_L}^{P-p_1} {N_R}^{Q-q} {N_C}^{1-p_2}  \\
    &\leq \left(f(P,Q)\right)^2 {N_L}^{P} {N_R}^{Q} N_C,
\end{align}
where $f(P,Q)$ does not depend on $N_L, N_C, N_R$.
\end{proof}

Now we are ready to prove \cref{lem:trotter_frobenius}, which we restate here.
\begin{lemma}
\label{lem:trotter_frobenius_app}
Consider a graph with a vertex bottleneck. Let $H$ be any time-independent Hamiltonian that respects this connectivity.
    Let $\tU$ be the architecture-respecting simulation circuit corresponding to dividing the time-evolution $U=U(H,t)$ into $M$ equal segments, and simulating each by the $(2k)^{\mathrm{th}}$-order Trotter-Suzuki formula.
    Then, $\tU$ has depth $d = 2 \times 5^{k-1} M$ \cite{Trotter_Frob22}
    and there exists a function $g(k)$, only dependent on $k$, such that
    \begin{equation}
    \label{eq:U-tU<comm_app}
    \norm{U- \tU}_{\rm F} \le g(k) \frac{t^{2k+1}}{M^{2k}} \sqrt{{N_L}^{2k} N_R N_C}.
    \end{equation}
\end{lemma}
\begin{proof}
    We use the decomposition in \cref{eq:H=LR} to define $\tU$ as the simulation circuit corresponding to dividing the time-evolution by $H$, $U(H,t)$, into $M$ equal segments, and simulating each by the $(2k)$th-order Trotter-Suzuki formula. For example, if $k=1$, 
\begin{equation}
    \tU = \lr{e^{-\ii (t/2M) H_{LC}}e^{-\ii (t/M) H_R}e^{-\ii (t/2M) H_{LC}}}^M.
\end{equation}
$\tU$ consists of $ 2 \times 5^{k-1}M$ gates each acting on either the left and center qubit, or the right and center qubit  \cite{Trotter_Frob22}. We call $\tU$ the circuit approximation.
According to Theorem 2 in Ref.~\cite{Trotter_Frob22} (but using the normalized Frobenius norm), 
\begin{align}
    \norm{U- \tU}_{\rm F} &\le 2(2\times 5^{k-1})^{2k+1} M \lr{\frac{t}{M}}^{2k+1} \sum_{\eta_1,\cdots,\eta_{2k+1}\in \{L,R\} } \norm{[H_{\eta_{2k+1}}, \cdots, [H_{\eta_2}, H_{\eta_1}]]}_{\rm F}, \nonumber\\
    &\le 2(4\times 5^{k-1})^{2k+1} M \lr{\frac{t}{M}}^{2k+1} \max_{\eta_1,\cdots,\eta_{2k+1}\in \{L,R\} } \norm{[H_{\eta_{2k+1}}, \cdots, [H_{\eta_2}, H_{\eta_1}]]}_{\rm F}.
\end{align}

Let $f^*_{2k+1} \coloneqq \max_{P,Q} f(P,Q) \textrm{ s.t.} \, P+Q = 2k+1$. 
Then \cref{lem:H<sqrtN_app} implies the distance is bounded by
\begin{equation}\label{eq:frob_circuit}
    \norm{U- \tU}_{\rm F} \le 2(4\times 5^{k-1})^{2k+1} M \lr{\frac{t}{M}}^{2k+1}
    \sqrt{f^*_{2k+1}} \sqrt{{N_L}^{2k} N_R N_C},
\end{equation}
which completes the proof with $g(k) = 2(4\times 5^{k-1})^{2k+1}\sqrt{f^*_{2k+1}}$.
\end{proof}
\end{document}